\documentclass[journal]{IEEEtran}

\usepackage{graphicx}
\usepackage{placeins}
\usepackage{amsmath,amsthm,amsfonts,amssymb}
\usepackage{cite}
\usepackage{bm}
\usepackage{url}
\usepackage{array}
\usepackage{enumitem}
\usepackage{algorithm}
\usepackage[noend]{algpseudocode}
\usepackage[english]{babel}
\graphicspath{{figs/}}
\usepackage{xcolor}

\theoremstyle{plain}
\newtheorem{thm}{Theorem}
\newtheorem{lem}[thm]{Lemma}
\newtheorem{prop}[thm]{Proposition}
\newtheorem{cor}[thm]{Corollary}
\newtheorem{defi}{Definition}
\newtheorem{rem}{Remark}
\newtheorem{example}{Example}

\newenvironment{NewProof}{{\noindent\it Proof.}}{\hfill $\blacksquare$\par}

\newcommand{\C}{\mathbb C}
\newcommand{\R}{\mathbb R}
\newcommand{\E}{\mathbb E}
\newcommand{\CN}{\mathcal{CN}}
\newcommand{\cP}{\mathcal P}
\newcommand{\cR}{\mathcal R}
\newcommand{\cQ}{\mathcal Q}
\newcommand{\cC}{\mathcal C}
\newcommand{\cS}{\mathcal S}
\newcommand{\tr}{\operatorname{tr}}
\newcommand{\diag}{\operatorname{diag}}
\newcommand{\blkdiag}{\operatorname{blkdiag}}
\newcommand{\rank}{\operatorname{rank}}
\newcommand{\dd}{\mathrm d}
\newcommand{\argmax}{\operatorname*{arg\,max}}

\begin{document}

\title{Networked Embodied Communication: From Collective Distinguishability to Communication Reliability}

\author{Yewen Cao and Yulin Shao%
\thanks{The authors are with the Department of Electrical and Computer Engineering, The University of Hong Kong, Hong Kong, China (e-mails: caoyewen3@gmail.com, ylshao@hku.hk).}%
}

\maketitle

\begin{abstract}
Embodied agents need to convey information to surrounding infrastructure, but their active communication interfaces may be unavailable, constrained, or intentionally inactive.
Their ability to manipulate physical states offers a complementary path: messages can be encoded in deliberately selected configurations and recovered through infrastructure sensing. 
This principle underlies embodied communication. 
Yet physical differences do not guarantee distinguishable messages: a single sensing viewpoint may leave ambiguities that repeated sensing cannot resolve. 
This paper develops networked embodied communication, where distributed access points (APs) jointly observe message-bearing scatterer positions under fixed illumination.
Under a correlated Gaussian sensing model, we characterize the additional distinguishability supplied by receive APs, establish exact redundancy conditions, and reveal how distinctions absent from individual observations can emerge through cross-AP statistical relationships. 
We then establish the exact asymptotic optimal maximum-error behavior of a finite alphabet under repeated independent sensing. 
The largest group of indistinguishable messages determines the error floor; once all messages are distinguishable, the minimum pairwise Chernoff information determines the error exponent. 
For a given alphabet, receiver cooperation can therefore eliminate an error floor that repetition at any individual AP cannot overcome. 
Building on these results, we derive finite-budget reliability conditions and jointly design the receive AP set and message-bearing positions. 
Numerical results show that the proposed search closely approaches exact benchmarks on reduced instances with substantially fewer candidate evaluations than exhaustive enumeration, while receiver cooperation reduces the sensing intervals needed to guarantee reliable decoding.
\end{abstract}


\section{Introduction}
\label{sec:intro}

Embodied artificial intelligence (AI) is increasingly moving from standalone autonomous agents toward agents operating within sensing-enabled intelligent environments \cite{duan2022survey,cui2024llmind}. In such environments, embodied agents need to convey task- and coordination-related information to surrounding infrastructure and networked services \cite{shao2024theory}. Conventional wireless communication naturally supports this exchange through a radio frequency (RF) communication interface, but for some agents or operating conditions, such an interface may be unavailable \cite{sitti2015biomedical}, constrained \cite{lu2020batteryless}, intermittent \cite{saboia2022achord}, or intentionally inactive \cite{bash2013limits}. Meanwhile, the agent may still be able to deliberately alter its physical state or surroundings, while the infrastructure can sense the resulting physical changes. This suggests that a controllable physical state itself can serve as an information-bearing medium \cite{pagello1999cooperative}.

Embodied communication formalizes this principle by mapping messages onto deliberately selected physical states or actions and recovering the intended message from sensing observations \cite{shao2026embodied,dragan2013legibility,wu2025actions}. The physical world thereby becomes part of the communication mechanism, without requiring the agent to transmit a data-bearing waveform. Depending on the application, the information-bearing state may take the form of a configuration, motion, or interaction, while the sensing modality may be RF, visual, infrared, acoustic, or multimodal.

The convergence of embodied communication and wireless systems is enabled by integrated sensing and communication (ISAC) \cite{liu2022isacsurvey}, which equips wireless infrastructure with RF sensing capabilities. The agent's ability to act and the infrastructure's ability to sense can therefore remain available even when the agent's RF communication interface is unavailable or inactive. Recent studies instantiate the information-bearing physical state through the position of a controllable scatterer \cite{shao2026embodied,lei2026nearfield}: the agent conveys a message by placing the scatterer at a deliberately selected position, while the wireless infrastructure illuminates the environment and infers that position from the resulting RF observations. These studies show that physically realizable scatterer positions can form an embodied alphabet. Crucially, communication reliability is determined not by the physical differences among these positions alone, but by the statistical distinguishability of the sensing observations they induce.

Existing studies, however, view the embodied state through a single sensing site. This places a fundamental limit on the statistical distinctions available for decoding: two physically different scatterer positions may induce only weakly separated observation distributions from one viewpoint, or may even induce the same distribution. In the former case, repeated sensing can accumulate evidence and improve reliability; in the latter, no amount of repetition under the same sensing configuration can reveal a distinction that is simply not present \cite{lei2026nearfield}. The limitation is therefore not always one of sensing duration, but of sensing diversity. This motivates \emph{networked embodied communication}, in which spatially distributed receivers observe the same physical state from different viewpoints and their observations are jointly processed, allowing distinctions that are missing locally to emerge collectively.

Cell-free ISAC provides a natural foundation for realizing networked embodied communication \cite{demirhan2024cellfree,mao2024communicationsensing}. Spatially distributed access points (APs) can observe the same physical state from different locations and forward their measurements to a central processing unit (CPU) for joint processing \cite{ngo2017cellfree,bjornson2020competitive,behdad2024multistatic}. Distinguishability is then determined by the observations available to the receive network as a whole rather than by any single sensing viewpoint. This networked setting raises three challenges.
\begin{itemize}[leftmargin=0.45cm]

\item The value of a receive AP depends on what the network already observes. An AP that is informative on its own may contribute little when its information is redundant, while another AP may be valuable because it complements the existing viewpoints. Networked embodied communication therefore requires a notion of distinguishability that captures the contribution of observations collectively rather than in isolation.

\item Collective distinguishability must ultimately be translated into communication reliability. Some ambiguities can be reduced by accumulating more sensing observations, whereas others persist because the selected network does not provide the statistical distinction needed for decoding. A reliability theory is therefore needed to determine when network cooperation merely strengthens communication and when it fundamentally changes whether reliable communication is possible.

\item The receive network and the embodied alphabet are intrinsically coupled. The alphabet determines which physical states must be distinguished, while the selected APs determine which distinctions the network can provide. Reliable networked embodied communication therefore requires the sensing infrastructure and the message-bearing physical states to be designed together. This coupling distinguishes networked embodied communication from sensor selection problems with prescribed hypotheses \cite{bajovic2011sensor,shao2021federated} and conventional receiver selection problems in sensing and ISAC \cite{yan2024communicate,wang2026receiver}.
\end{itemize}

To address these challenges, this paper develops a networked embodied communication framework with centralized cell-free sensing. The transmit AP configuration and probing signals are kept fixed, while the participating receive APs are selected from a distributed candidate set. This isolates the role of receive-side sensing diversity and allows us to study how the network observations themselves shape the distinguishability, reliability, and design of embodied communication.

Our main contributions are summarized as follows.

\begin{enumerate}[leftmargin=0.45cm]

\item We establish a theory of collective distinguishability for networked embodied communication. Unlike single-viewpoint formulations, where distinguishability is determined by one sensing observation, the networked setting allows different receive APs to contribute complementary statistical distinctions. We characterize the additional pairwise distinguishability brought by new AP observations and derive an exact condition for when those observations are redundant. The analysis further reveals that, with correlated observations, a distinction can arise from the joint statistical relationship among APs rather than from the standalone observation of any individual AP. This provides a basis for understanding when distributed sensing creates genuinely new message distinctions rather than merely collecting more measurements.

\item We establish the connection between collective distinguishability and communication reliability. For a finite embodied alphabet observed over repeated sensing intervals, we show that unresolved statistical ambiguities and weak but nonzero distinctions lead to fundamentally different reliability behaviors: the former produce an exact asymptotic maximum-error floor, while the latter permit vanishing error with an exponent determined by the weakest symbol pair. This characterization shows that receive AP cooperation can do more than improve decoding performance; it can change whether arbitrarily reliable embodied communication is possible at all. We further extend the analysis to finite sensing budgets by deriving sufficient and necessary reliability conditions and translating them into guaranteed communication rates normalized by sensing time.

\item We develop a reliability-driven framework for jointly designing the receive network and the embodied alphabet. Because the alphabet determines which physical state distinctions must be supported and the AP set determines which distinctions are available, neither can be optimized independently of the other. Guided by the reliability theory, our design first resolves statistical ambiguities and then strengthens the weakest remaining distinctions. We show that receive AP selection exhibits complementarity rather than the standard diminishing-returns structure, and develop an efficient joint search supported by a reduced-dimensional Chernoff representation. Numerical results demonstrate that the proposed design closely approaches exact discrete benchmarks on reduced instances with substantially fewer candidate evaluations, while the resulting network cooperation reduces the sensing resources required to support reliable embodied communication.

\end{enumerate}

\begin{figure}[!t]
\centering
\includegraphics[width=0.9\linewidth]{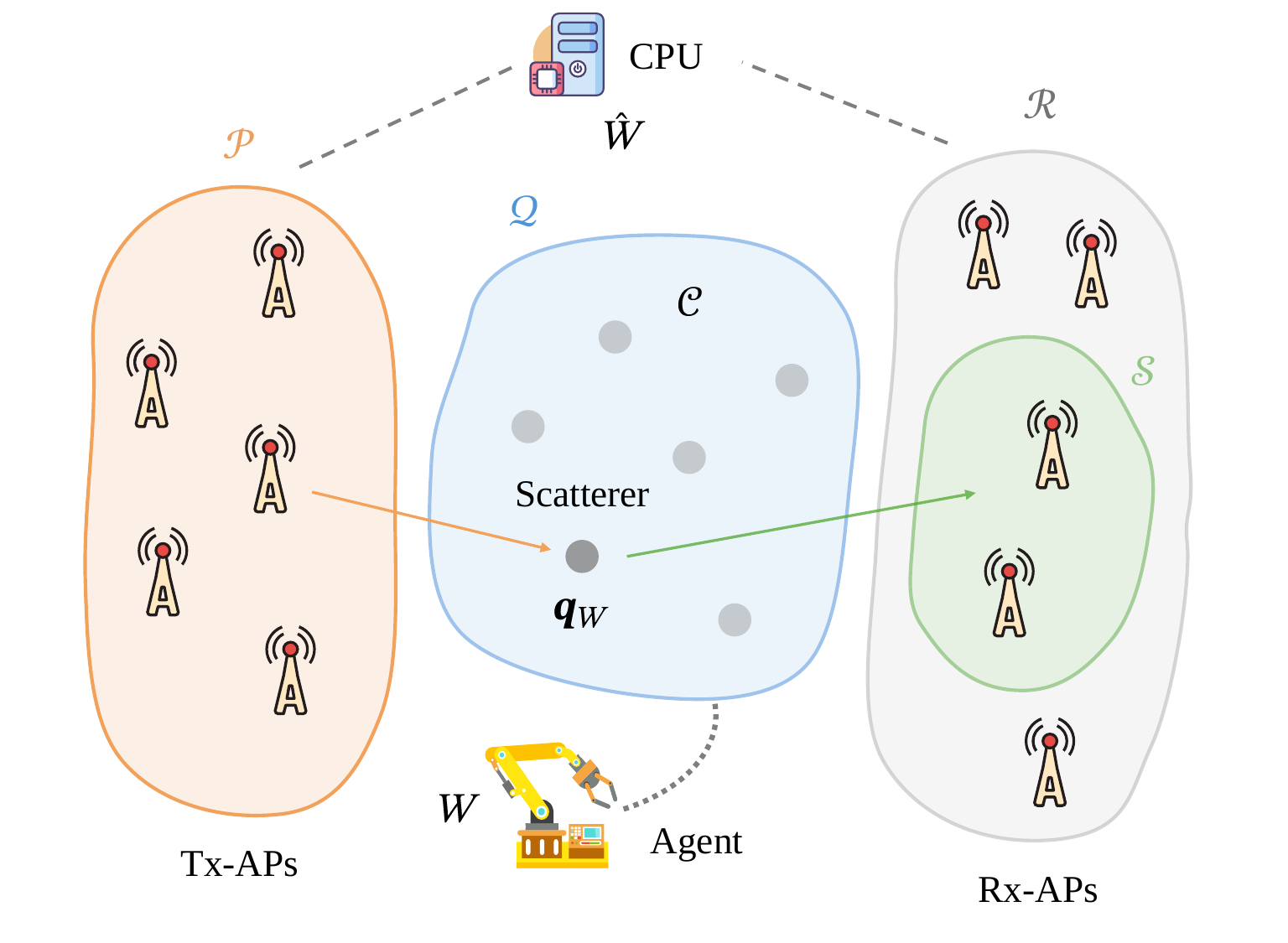}
\caption{Networked embodied communication through the position of a controllable scatterer.
The agent maps message $W$ to scatterer position $\bm q_W$.
The transmit APs illuminate the scatterer, the selected receive APs observe the resulting bistatic echoes, and the CPU jointly processes the distributed observations to recover $W$.}
\label{fig:system}
\end{figure}

\emph{Notation:} Lowercase and uppercase boldface letters denote vectors and matrices, respectively, and calligraphic letters denote sets.
The superscripts $(\cdot)^T$ and $(\cdot)^H$ denote transpose and Hermitian transpose, while $\bm0$, $\bm1$, and $\bm I$ denote an all-zero vector or matrix, an all-one vector, and an identity matrix, respectively, with dimensions clear from context.
The operators $\diag(\cdot)$ and $\blkdiag(\cdot)$ denote diagonal extraction and block-diagonal formation, respectively, and $\CN(\bm\mu,\bm\Sigma)$ denotes a proper complex Gaussian distribution with mean $\bm\mu$ and covariance $\bm\Sigma$.

\section{System Model}
\label{sec:model}

We consider the networked embodied communication system illustrated in Fig.~\ref{fig:system}. The system consists of a CPU, a set of distributed APs, and an embodied agent. The APs form a centralized cell-free ISAC network and are connected to the CPU through fronthaul links. They are assumed fully synchronized for centralized sensing \cite{behdad2024multistatic}.

The embodied agent can deliberately reposition a controllable scatterer within its accessible region. This physical degree of freedom is used for communication. Let $W\in\{1,\ldots,J\}$ denote an equiprobable message index, and let $\cC=\{\bm q_1,\ldots,\bm q_J\}$ denote the embodied alphabet, whose elements are admissible scatterer positions. To convey message $W=j$, the agent moves the scatterer to $\bm q_j$, allows it to settle, and maintains it there while the message is sensed. Each $\bm q_j$ is an embodied symbol, physically realized as the scatterer position selected by the agent. The message is therefore instantiated in the environment rather than encoded into a data-bearing RF waveform generated by the agent \cite{shao2026embodied}.

The transmit APs illuminate the scatterer, while selected receive APs observe its bistatic echoes from different viewpoints. Their observations are forwarded to the CPU, which jointly processes them to recover $W$. Thus, the information flow is $W
\to
\bm q_W
\to
\text{distributed sensing observations}
\to
\widehat W$.

\subsection{Cell-Free Sensing Infrastructure}

Our focus is how distributed receive observations change the distinguishability of embodied symbols, in particular whether additional AP observations can resolve ambiguities that persist under a given sensing illumination. Accordingly, we keep the transmit configuration fixed and treat receive AP participation as the network design variable. Let $\cP$ and $\cR$ denote the transmit AP set and the candidate receive AP set, respectively, with $\cP\cap\cR=\emptyset$. The transmit AP set $\cP$ and the probing signals remain fixed throughout the paper, while $\cS\subseteq\cR$ denotes the receive AP set selected for sensing. Changing $\cS$ therefore changes the observations available to the CPU without changing the illumination that induces them.

We use a common Cartesian coordinate system for the APs and the region accessible to the agent \cite{zhang2025polarization}. The controllable scatterer can be placed within a bounded horizontal region $\cQ\subset\R^3$ at height $h_{\rm A}$, so any $\bm q\in\cQ$ can be written as $\bm q=[q_x,q_y,h_{\rm A}]^T$. Accordingly, $\cC\subseteq\cQ$.

AP $a\in\cP\cup\cR$ is located at $\bm b_a\in\R^3$ and is equipped with an $M_1\times M_2$ uniform planar array (UPA), comprising $M=M_1M_2$ antennas with half-wavelength spacing. Let $\bm e_{a,1}$ and $\bm e_{a,2}$ denote orthonormal unit vectors along its local array axes. For $d\in\{1,2\}$, the directional cosine from AP $a$ to scatterer position $\bm q$ along its $d$th local array axis is $\mu_{a,d}(\bm q)=\bm e_{a,d}^T(\bm q-\bm b_a)/\|\bm q-\bm b_a\|$.
The normalized far-field steering vector is
\[
\bm a_a(\bm q)
=
\bm a_1\bigl(\mu_{a,1}(\bm q)\bigr)
\otimes
\bm a_2\bigl(\mu_{a,2}(\bm q)\bigr),
\]
where
\[
\bm a_d(\mu)
=
\frac{1}{\sqrt{M_d}}
\bigl[
1,e^{\jmath\pi\mu},\ldots,e^{\jmath\pi(M_d-1)\mu}
\bigr]^T.
\]

Downlink communication and sensing occupy separate intervals \cite{liu2022isacsurvey,zhang2022timeDivision}. An active sensing interval begins after the scatterer has settled at $\bm q_W$, contains $\tau_s$ probing symbols of duration $T_s$, and lasts $T_{\rm sen}=\tau_sT_s$. The scatterer remains at $\bm q_W$ throughout the interval. At symbol index $m\in\{1,\ldots,\tau_s\}$, transmit AP $t\in\cP$ sends the predetermined probing vector $\bm x_t[m]\in\C^M$, which is known to the CPU. The same probing sequence is reused across the $L$ sensing intervals used to decode one message, so $m$ indexes probing symbols within an interval, whereas $L$ counts repeated observations of the same embodied symbol.

\subsection{Multistatic Observation Model}

For embodied symbol $\bm q_j$, under the narrowband far-field multistatic sensing model \cite{behdad2024multistatic}, the bistatic response from transmit AP $t$ through the controllable scatterer to receive AP $r$ at symbol index $m$ is
\[
\bm g_{r,t,j}[m]
=
\sqrt{\beta_{r,t,j}}\,
\bm a_r(\bm q_j)
\bm a_t^T(\bm q_j)
\bm x_t[m],
\]
where $\beta_{r,t,j}\ge0$ accounts for the bistatic propagation loss and average scattering power. The remaining scattering amplitude and phase are represented by $\alpha_{r,t,j}$.

Using a fixed ordering of $\cP$, define $\bm\alpha_{r,j}=[\alpha_{r,t,j}]_{t\in\cP}\in\C^{|\cP|}$ and model $\bm\alpha_{r,j}\sim\CN(\bm0,\bm R_{\alpha,r,j})$ with $\diag(\bm R_{\alpha,r,j})=\bm1$. The covariance $\bm R_{\alpha,r,j}$ describes correlations among the normalized bistatic scattering coefficients and may vary with $\bm q_j$. The coefficients remain constant during one sensing interval \cite{swerling1960probability}. Let $\bm G_{r,j}[m]\in\C^{M\times|\cP|}$ collect $\bm g_{r,t,j}[m]$, $t\in\cP$, as columns in the same order.

The probing signals also reach receive AP $r$ through direct propagation and reflections from uncontrolled objects. Following the residual channel treatment in cell-free multistatic sensing \cite{behdad2024multistatic}, propagation components treated as known are canceled at the sensing front end. The remaining propagation from transmit AP $t$ to receive AP $r$ is represented by $\bm H_{r,t}\in\C^{M\times M}$. The residual channels are modeled as jointly zero-mean complex Gaussian and may be correlated across AP pairs. Their realizations remain constant during one sensing interval and are not assumed to be known instantaneously at the CPU.

After cancellation of the known propagation components, receive AP $r$ observes $\bm y_r[m]\in\C^M$ as
\begin{equation}\label{eq:sig2}
\bm y_r[m]
=
\sum_{t\in\cP}
\alpha_{r,t,j}\bm g_{r,t,j}[m]
+
\sum_{t\in\cP}
\bm H_{r,t}\bm x_t[m]
+
\bm n_r[m].
\end{equation}
The first term is the echo from the controllable scatterer at $\bm q_j$. The receiver noise satisfies $\bm n_r[m]\sim\CN(\bm0,\sigma_{n,r}^2\bm I_M)$, where $\sigma_{n,r}^2$ denotes the noise variance at receive AP $r$. Noise samples are independent across receive APs and probing symbols. They are also independent of the scattering coefficients and residual channels.

The second and third terms of \eqref{eq:sig2}  do not depend on $\bm q_j$, but they affect how reliably the induced echoes can be distinguished. Since $\bm H_{r,t}$ remains constant over one sensing interval, its contribution can correlate the residual observations across probing symbols. Correlation across receive APs is also allowed and will be retained in the network covariance below.

Define
\[
\bm w_r[m]
=
\sum_{t\in\cP}
\bm H_{r,t}\bm x_t[m]
+
\bm n_r[m].
\]
Over one sensing interval, collect the observations associated with the $\tau_s$ probing symbols at receive AP $r$ as $\bm y_r=[\bm y_r^T[1],\ldots,\bm y_r^T[\tau_s]]^T$ and the corresponding residual terms as $\bm w_r=[\bm w_r^T[1],\ldots,\bm w_r^T[\tau_s]]^T$. Let $D=\tau_sM$ denote the observation dimension at one receive AP. Define
\[
\bm G_{r,j}
=
\bigl[
\bm G_{r,j}^T[1],\ldots,\bm G_{r,j}^T[\tau_s]
\bigr]^T
\in\C^{D\times|\cP|}.
\]
The observation at receive AP $r$ over one sensing interval is then $\bm y_r=\bm G_{r,j}\bm\alpha_{r,j}+\bm w_r$.
We consider sensing intervals sufficiently separated relative to the coherence times of the scattering and residual channels. Their realizations are therefore modeled as independent across sensing intervals with unchanged covariance matrices.

To compare different receive AP selections under the same transmit configuration, fix an ordering of the candidate set $\cR$ and form $\bm y_{\cR}=[\bm y_r^T]_{r\in\cR}^T$. Using the same AP ordering, let $\bm\alpha_{\cR,j}$ and $\bm w_{\cR}$ collect the corresponding scattering and residual vectors, and define $\bm G_{\cR,j}=\blkdiag\{\bm G_{r,j}:r\in\cR\}$. The observation over the full candidate set is $\bm y_{\cR}=\bm G_{\cR,j}\bm\alpha_{\cR,j}+\bm w_{\cR}$.
Here, $\bm\alpha_{\cR,j}\sim\CN(\bm0,\bm R_{\alpha,\cR,j})$ and $\bm w_{\cR}\sim\CN(\bm0,\bm K(\cR))$, where $\bm K(\cR)\succ\bm0$ contains the contributions of the residual channels and receiver noise. Off-diagonal blocks of $\bm R_{\alpha,\cR,j}$ and $\bm K(\cR)$ describe correlations across receive APs, and $\bm\alpha_{\cR,j}$ and $\bm w_{\cR}$ are independent. If the scattering and residual terms are independent across receive APs, the corresponding off-diagonal blocks vanish.

For $\cS\subseteq\cR$, let $\bm P_{\cS}$ extract the observation blocks indexed by $\cS$, so that $\bm y_{\cS}=\bm P_{\cS}\bm y_{\cR}$. The resulting CPU observation has dimension $D_{\cS}\triangleq D|\cS|$. Using the same AP ordering, define $\bm G_{\cS,j}=\blkdiag\{\bm G_{r,j}:r\in\cS\}$, and let $\bm\alpha_{\cS,j}$ and $\bm w_{\cS}$ collect the corresponding scattering and residual vectors. They satisfy $\bm\alpha_{\cS,j}\sim\CN(\bm0,\bm R_{\alpha,\cS,j})$ and $\bm w_{\cS}\sim\CN(\bm0,\bm K(\cS))$, where $\bm R_{\alpha,\cS,j}$ is the corresponding principal block submatrix of $\bm R_{\alpha,\cR,j}$ and $\bm K(\cS)=\bm P_{\cS}\bm K(\cR)\bm P_{\cS}^H$. The CPU observation under receive AP set $\cS$ is therefore
\begin{equation}
\bm y_{\cS}
=
\bm G_{\cS,j}\bm\alpha_{\cS,j}
+
\bm w_{\cS}.
\label{eq:selected_network_observation}
\end{equation}

\subsection{Network Observation Statistics}

For a given embodied symbol, the CPU observes a random sensing vector rather than a deterministic received signal. Marginalizing the Gaussian scattering and residual terms in \eqref{eq:selected_network_observation} gives
\begin{equation}
\begin{aligned}
\bm y_{\cS}\mid W=j
&\sim
\CN\bigl(\bm0,\bm\Sigma_j(\cS)\bigr),\\
\bm\Sigma_j(\cS)
&=
\bm G_{\cS,j}
\bm R_{\alpha,\cS,j}
\bm G_{\cS,j}^H
+
\bm K(\cS).
\end{aligned}
\label{eq:block_cov}
\end{equation}
Let $p_j^{\cS}(\bm y)$ denote the density of $\bm y_{\cS}$ under $W=j$. When the receive AP set is clear, we simply write $p_j(\bm y)$. The first term in $\bm\Sigma_j(\cS)$ depends on $\bm q_j$, whereas $\bm K(\cS)$ is common to all embodied symbols. The CPU is assumed to know $\bm\Sigma_j(\cR)$ for each candidate scatterer position, for example from the statistical model, calibration, or measurements collected over time.

Across the repeated sensing intervals, the independence assumption above gives conditionally i.i.d. observations with density $p_j^{\cS}$. For $\cS$, the embodied alphabet $\cC$ therefore induces the family of observation distributions $\{p_j^{\cS}\}_{j=1}^{J}$ at the CPU. Reliable decoding depends first on whether different embodied symbols induce distinct distributions and, when they do, on how strongly these distributions are separated. Changing $\cS$ does not change the positions in $\cC$, but it can change the statistical distinctions available to the CPU.

Receive AP selection retains the corresponding blocks of the reference observation over $\cR$. Hence, $\bm\Sigma_j(\cS)=\bm P_{\cS}\bm\Sigma_j(\cR)\bm P_{\cS}^H$. For $\cS_1\subseteq\cS_2$, the observation under $\cS_1$ is therefore a marginal of the observation under $\cS_2$. This nested structure allows us to characterize the additional distinguishability contributed by new receive APs relative to the observations already available under $\cS_1$.

\section{Collective Distinguishability under Distributed Sensing}
\label{sec:distinguishability}

The embodied communication framework in \cite{shao2026embodied} interprets sensing resolution through the distinguishability of intentionally selected physical states. In its RF realization, the pairwise distinction between two embodied symbols is characterized by the Bhattacharyya distance between their induced sensing distributions. With distributed sensing, the CPU observes each embodied symbol through multiple receive APs, giving a joint Gaussian distribution characterized by $\bm\Sigma_j(\cS)$. The relevant statistical distinction is therefore determined by these joint observation distributions and the selected receive AP set.

Chernoff information provides a natural measure for characterizing this pairwise statistical separation. Its optimizing parameter is determined by the statistical structure of the two distributions. This formulation allows the role of the Bhattacharyya distance in the RF model of \cite{shao2026embodied} to be examined within the more general distributed sensing model, and provides the pairwise characterization needed to study distinguishability across the embodied alphabet.

\subsection{Pairwise Chernoff Information}
\label{subsec:pairwise_chernoff_information}

Fix a receive AP set $\cS$ and consider embodied symbols $\bm q_i$ and $\bm q_j$. For brevity, write $\bm y=\bm y_{\cS}\in\C^{D_{\cS}}$ and $\bm\Sigma_k=\bm\Sigma_k(\cS)$, $k\in\{i,j\}$. The corresponding hypotheses are $\mathcal H_i:\bm y\sim\CN(\bm0,\bm\Sigma_i)$ and $\mathcal H_j:\bm y\sim\CN(\bm0,\bm\Sigma_j)$.
Although the CPU ultimately decodes among all $J$ embodied symbols, the induced binary problem isolates the intrinsic statistical separation of one pair independently of the remaining symbols. For equal priors, its minimum binary error probability is
\[
P_{e,ij}^{\star}
=
\frac{1}{2}
\int
\min\{p_i(\bm y),p_j(\bm y)\}
\dd\bm y.
\]

Chernoff information characterizes the overlap between the two induced distributions. It yields an exponential upper bound on the binary error probability and, under independent repetitions of the same embodied symbol, gives the exact asymptotic binary error exponent \cite{poor1994introduction}.

\begin{defi}[Chernoff Information~\cite{chernoff1952measure}]
For $s\in[0,1]$, define
\begin{equation*}
\begin{aligned}
C_{ij}(s;\cS)
&=
-\log
\int
p_i^s(\bm y)
p_j^{1-s}(\bm y)
\dd\bm y,\\
C_{ij}^{\star}(\cS)
&=
\max_{0\le s\le1}
C_{ij}(s;\cS).
\end{aligned}
\end{equation*}
\end{defi}

When $\cS$ is fixed, we write $C_{ij}(s)$ and $C_{ij}^{\star}$ for brevity. Since $\min\{p_i(\bm y),p_j(\bm y)\}\le p_i^s(\bm y)p_j^{1-s}(\bm y)$, the Chernoff bound gives $P_{e,ij}^{\star}\le\frac{1}{2}\exp[-C_{ij}(s;\cS)]$ for every $s\in[0,1]$. Maximizing $C_{ij}(s;\cS)$ therefore minimizes this exponential upper bound over the Chernoff family and gives $P_{e,ij}^{\star}\le\frac{1}{2}\exp[-C_{ij}^{\star}(\cS)]$. Thus, optimizing $s$ is relevant not only to the asymptotic exponent but also to obtaining the tightest bound within this family.

\begin{lem}[Gaussian Chernoff Form~\cite{dogandzic2003chernoff}]
\label{lem:gaussian_chernoff}
Let $\lambda_1,\ldots,\lambda_{D_{\cS}}$ denote the generalized eigenvalues of $(\bm\Sigma_j,\bm\Sigma_i)$, equivalently the eigenvalues of $\bm\Sigma_i^{-1/2}\bm\Sigma_j\bm\Sigma_i^{-1/2}$. Then
\begin{equation}
\begin{aligned}
C_{ij}(s)
&=
\log\det\bigl((1-s)\bm\Sigma_i+s\bm\Sigma_j\bigr)\\
&\quad
-(1-s)\log\det\bm\Sigma_i
-s\log\det\bm\Sigma_j\\
&=
\sum_{d=1}^{D_{\cS}}
\left[
\log\bigl((1-s)+s\lambda_d\bigr)
-s\log\lambda_d
\right].
\end{aligned}
\label{eq:gaussian_chernoff}
\end{equation}
If $\bm\Sigma_i\ne\bm\Sigma_j$, $C_{ij}(s)$ is strictly concave on $[0,1]$, and its unique maximizer is the root of
\begin{equation*}
\sum_{d=1}^{D_{\cS}}
\left[
\frac{\lambda_d-1}{(1-s)+s\lambda_d}
-\log\lambda_d
\right]
=0.
\end{equation*}
\end{lem}

The generalized eigenvalues characterize the covariance mismatch between the two hypotheses \cite{li2018topological}. A value $\lambda_d=1$ means that the corresponding generalized eigenmode has identical second order statistics under the two hypotheses and therefore contributes no Chernoff information. Only modes with $\lambda_d\ne1$ contribute to pairwise distinguishability. For distinct distributions, the strict concavity in Lemma~\ref{lem:gaussian_chernoff} reduces the optimization of the Chernoff parameter to a one dimensional root search \cite{nielsen2022revisiting}.

\begin{rem}[Relation to the Bhattacharyya Distance in \cite{shao2026embodied}]
The Bhattacharyya distance is the midpoint value $B_{ij}=C_{ij}(1/2)$ \cite{bhattacharyya1943measure}. For the sensing model with a single viewpoint in \cite{shao2026embodied}, $\bm\Sigma_j=\sigma^2(\bm I+\gamma\bm a_j\bm a_j^H)$ with $\|\bm a_j\|=1$, where the effective sensing SNR $\gamma$ is common to all embodied symbols. With $\eta_{ij}=|\bm a_i^H\bm a_j|^2$, Lemma~\ref{lem:gaussian_chernoff} reduces to
\[
C_{ij}(s)
=
\log
\left[
1+
\frac{\gamma^2}{1+\gamma}
s(1-s)
(1-\eta_{ij})
\right].
\]
The common effective SNR and normalized rank one covariance structure make this expression symmetric about $s=1/2$. Hence, $C_{ij}^{\star}=C_{ij}(1/2)=B_{ij}$, which recovers the Bhattacharyya distinguishability measure used in \cite{shao2026embodied}.

The general distributed sensing covariance in \eqref{eq:block_cov} does not retain this symmetry because the bistatic gains, array responses, probing illumination, and scattering correlations can vary with the embodied symbol. Therefore, $C_{ij}^{\star}(\cS)\ge C_{ij}(1/2;\cS)=B_{ij}(\cS)$, and optimizing $s$ gives a Chernoff bound no looser than the Bhattacharyya bound, with a strict improvement whenever the unique maximizer satisfies $s^{\star}\ne1/2$.
\end{rem}

For the Gaussian laws with positive definite covariance matrices considered here, $C_{ij}^{\star}(\cS)=0\Longleftrightarrow p_i^{\cS}=p_j^{\cS}\Longleftrightarrow\bm\Sigma_i(\cS)=\bm\Sigma_j(\cS)$. Thus, physically different scatterer positions can induce the same observation distribution under a given receive AP set. Zero Chernoff information marks a missing statistical distinction under $\cS$, whereas a positive value quantifies a distinction that is present. Receive AP complementarity is therefore determined by whether additional observations resolve an unresolved pair or strengthen a distinction already available at the CPU.

\subsection{Receive AP Complementarity}
\label{subsec:receive_ap_complementarity}

To record which distinctions remain unresolved under a receive AP set, fix $\cC$ and $\cS$ and write $i\sim_{\cS}j$ when symbols $i$ and $j$ are statistically indistinguishable under $\cS$. We call such symbols \emph{equivalent} under $\cS$. Let $\Pi_{\cS}$ denote the resulting equivalence partition of $\{1,\ldots,J\}$. The embodied alphabet $\cC$ is called \emph{identifiable} under $\cS$ if every equivalence class in $\Pi_{\cS}$ is a singleton.

\begin{example}[Complementary Receive APs]
\label{ex:complementary_viewpoints}
Consider the three embodied symbols illustrated in Fig.~\ref{fig:ap_complementarity}. Receive AP $r_1$ cannot distinguish $\bm q_2$ from $\bm q_3$, while receive AP $r_2$ cannot distinguish $\bm q_1$ from $\bm q_2$. Thus,
\[
\Pi_{\{r_1\}}
=
\bigl\{
\{1\},\{2,3\}
\bigr\},
\qquad
\Pi_{\{r_2\}}
=
\bigl\{
\{1,2\},\{3\}
\bigr\},
\]
whereas their joint observation gives
\[
\Pi_{\{r_1,r_2\}}
=
\bigl\{
\{1\},\{2\},\{3\}
\bigr\}.
\]
Neither AP identifies the alphabet alone, but their observations identify it collectively.
\end{example}

\begin{figure}[t]
\centering
\includegraphics[width=0.80\linewidth]{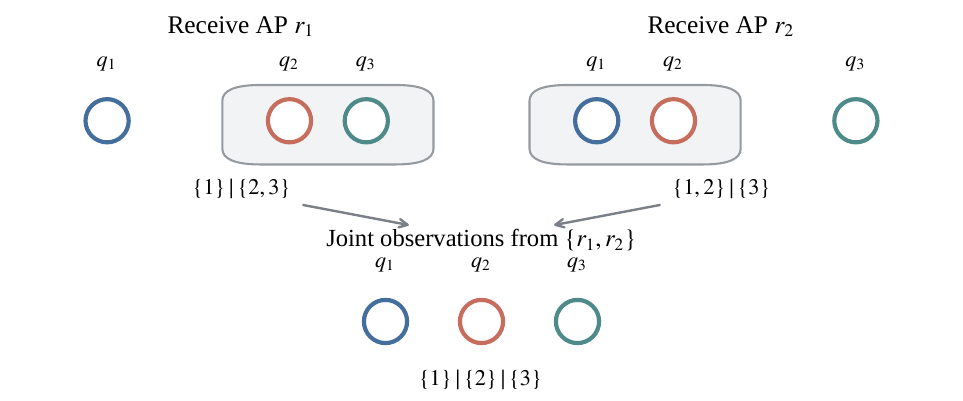}
\caption{Receive AP complementarity. Each AP leaves a different symbol pair statistically indistinguishable, whereas their joint observations make the embodied alphabet identifiable.}
\label{fig:ap_complementarity}
\end{figure}

Example~\ref{ex:complementary_viewpoints} shows that receive APs can provide complementary distinctions: neither AP identifies the embodied alphabet alone, yet each resolves an ambiguity left by the other. More generally, the value of an additional receive AP depends on what statistical distinction it contributes beyond the observations already available at the CPU.

This incremental contribution can be characterized because nested receive AP sets correspond to marginals of the same network observation model established in Section~\ref{sec:model}. Consider nonempty receive AP sets $\emptyset\ne\cS_1\subset\cS_2$. Let $\bm y$ collect the observations over $\cS_1$, and let $\bm z$ collect the additional observations from $\cS_2\setminus\cS_1$. The relevant question is whether $\bm z$ provides a statistical distinction between $\bm q_i$ and $\bm q_j$ beyond that already contained in $\bm y$.

Under hypothesis $\mathcal H_k$, $k\in\{i,j\}$, partition the joint covariance of $(\bm y,\bm z)$ as
\[
\bm\Sigma_k(\cS_2)
=
\begin{bmatrix}
\bm\Sigma_{y,k} & \bm\Sigma_{yz,k}\\
\bm\Sigma_{zy,k} & \bm\Sigma_{z,k}
\end{bmatrix}.
\]
Here, $\bm\Sigma_{y,k}$ and $\bm\Sigma_{z,k}$ are the marginal covariances of the existing and additional observations, while $\bm\Sigma_{zy,k}=\bm\Sigma_{yz,k}^H$ is their cross covariance. Physically, this cross covariance describes the statistical fluctuations shared across the two receive AP groups under embodied symbol $k$, which can arise from correlated scattering and residual propagation components.

Although the joint observation has zero mean under both hypotheses, conditioning the additional observations on those already available gives
\begin{equation}
\bm z\mid\bm y,\mathcal H_k
\sim
\CN(\bm A_k\bm y,\bm V_k),
\label{eq:conditional_innovation_model}
\end{equation}
where $\bm A_k=\bm\Sigma_{zy,k}\bm\Sigma_{y,k}^{-1}$ and $\bm V_k=\bm\Sigma_{z,k}-\bm A_k\bm\Sigma_{yz,k}$.
For the jointly Gaussian model, $\bm A_k\bm y$ is the conditional prediction of the additional observations from those already available at the CPU, and $\bm V_k$ is the covariance of the remaining innovation after this predictable component is removed. Hence, an embodied symbol can be distinguished not only through the marginal statistics at a new receive AP, but also through how that observation is statistically related to the observations already collected by the network.

Conditional Chernoff information has previously been used to characterize additional information from mixed observations in compressed hypothesis testing \cite{xu2013compressed}. Here, the observations are nested blocks from correlated receive APs. The following theorem quantifies the additional pairwise distinguishability contributed by these observations and gives the exact condition under which they are redundant.

\begin{thm}[Distinguishability Contribution of Additional Receive APs]
\label{thm:additional_ap_distinguishability}
Let $\emptyset\ne\cS_1\subset\cS_2$.

\emph{1) Chernoff function increment:}
For any embodied symbol pair $(i,j)$ and $s\in(0,1)$,
\begin{equation}
C_{ij}(s;\cS_2)-C_{ij}(s;\cS_1)
=
\Delta_{ij}^{\rm cov}(s)
+
\Delta_{ij}^{\rm mean}(s),
\label{eq:chernoff_function_increment}
\end{equation}
where
\begin{equation*}
\begin{aligned}
\Delta_{ij}^{\rm cov}(s)
&=
\log\det\overline{\bm V}_{ij}(s)
-(1-s)\log\det\bm V_i\\
&\quad
-s\log\det\bm V_j,\\
\Delta_{ij}^{\rm mean}(s)
&=
\log\det
\bigl(
\overline{\bm V}_{ij}(s)
+
\bm\Omega_{ij}(s)
\bigr)\\
&\quad
-\log\det\overline{\bm V}_{ij}(s),
\end{aligned}
\end{equation*}
with
\[
\begin{aligned}
\overline{\bm V}_{ij}(s)
&=(1-s)\bm V_i+s\bm V_j,\\
\bm Q_{ij}(s)
&=
\left[
s\bm\Sigma_{y,i}^{-1}
+
(1-s)\bm\Sigma_{y,j}^{-1}
\right]^{-1},\\
\bm\Omega_{ij}(s)
&=
s(1-s)
(\bm A_i-\bm A_j)
\bm Q_{ij}(s)
(\bm A_i-\bm A_j)^H.
\end{aligned}
\]

\emph{2) Redundancy condition:}
Both $\Delta_{ij}^{\rm cov}(s)$ and $\Delta_{ij}^{\rm mean}(s)$ are nonnegative, and
\begin{equation}
C_{ij}(s;\cS_2)=C_{ij}(s;\cS_1)
\quad\Longleftrightarrow\quad
\bm A_i=\bm A_j
\ \text{and}\
\bm V_i=\bm V_j.
\label{eq:conditional_redundancy_condition}
\end{equation}
Consequently,
\begin{equation}
C_{ij}^{\star}(\cS_2)
\ge
C_{ij}^{\star}(\cS_1),
\label{eq:pairwise_chernoff_monotonicity}
\end{equation}
with equality if and only if \eqref{eq:conditional_redundancy_condition} holds.

\emph{3) Identifiability:}
The equivalence partition $\Pi_{\cS_2}$ refines $\Pi_{\cS_1}$, and for any $1\le i<j\le J$,
\begin{equation}
i\sim_{\cS_2}j
\quad\Longleftrightarrow\quad
i\sim_{\cS_1}j,\ 
\bm A_i=\bm A_j,
\ \text{and}\
\bm V_i=\bm V_j.
\label{eq:identifiability_under_ap_expansion}
\end{equation}
Hence, $\cC$ is identifiable under $\cS_2$ if and only if every $i<j$ satisfying $i\sim_{\cS_1}j$ also satisfies $\bm A_i\ne\bm A_j$ or $\bm V_i\ne\bm V_j$.
\end{thm}

\emph{Proof:} See the supplementary material.

Theorem~\ref{thm:additional_ap_distinguishability} identifies two mechanisms that contribute to the Chernoff function increment from additional receive APs. The term $\Delta_{ij}^{\rm cov}(s)$ is positive when the conditional innovation covariance differs between the two embodied symbols. Although the original network observations are zero mean, $\Delta_{ij}^{\rm mean}(s)$ captures a difference between the conditional means $\bm A_i\bm y$ and $\bm A_j\bm y$ once the existing observations are given. This second contribution can arise from correlations across receive APs even when the marginal statistics of the added observations alone do not distinguish the two symbols.

The equality condition gives the corresponding criterion for redundancy. The added observations provide no new pairwise distinguishability if and only if $p_i(\bm z\mid\bm y)=p_j(\bm z\mid\bm y)$ for almost every $\bm y$, equivalently $\bm A_i=\bm A_j$ and $\bm V_i=\bm V_j$. Thus, the value of an additional receive AP is determined by whether its conditional observation law reveals a distinction beyond that already available at the CPU, rather than by its standalone discrimination alone. For a pair that is unresolved under $\cS_1$, any difference in this conditional law separates the pair under $\cS_2$, whereas repeated observations over $\cS_1$ cannot create such a distinction.

The decomposition in \eqref{eq:chernoff_function_increment} is exact for a fixed $s$. The optimized Chernoff information also accounts for a possible change in the maximizing parameter. If $C_{ij}^{\star}(\cS_2)>0$, let $s_2^{\star}\in(0,1)$ denote its unique maximizer. Then
\[
\begin{aligned}
C_{ij}^{\star}(\cS_2)
-
C_{ij}^{\star}(\cS_1)
&=
\Delta_{ij}^{\rm cov}(s_2^{\star})
+
\Delta_{ij}^{\rm mean}(s_2^{\star})\\
&\quad
-
\left[
C_{ij}^{\star}(\cS_1)
-
C_{ij}(s_2^{\star};\cS_1)
\right].
\end{aligned}
\]
The bracketed term is nonnegative because the Chernoff parameter is optimized separately for $\cS_1$. Thus, $\Delta_{ij}^{\rm cov}$ and $\Delta_{ij}^{\rm mean}$ exactly decompose the Chernoff function increment at a fixed parameter, while the increase in optimized Chernoff information also reflects the possible change in the maximizing $s$. The bracketed term vanishes when $s_2^{\star}$ also maximizes $C_{ij}(s;\cS_1)$. If $C_{ij}^{\star}(\cS_2)=0$, both hypotheses remain identical under $\cS_2$ and all these increments are zero.

When the receive AP observations are independent under every embodied symbol, the conditional structure simplifies. For disjoint receive AP subsets, $\bm\Sigma_{zy,k}=\bm0$, so $\bm A_k=\bm0$, $\bm V_k=\bm\Sigma_{z,k}$, and $\Delta_{ij}^{\rm mean}(s)=0$. The joint observation density also factorizes across receive APs. Hence, a symbol pair is distinguishable under $\cS$ if and only if at least one receive AP in $\cS$ distinguishes that pair.

\begin{cor}[Pair Coverage under Independent APs]
\label{cor:independent_identifiability}
Suppose the receive AP observations are independent under every embodied symbol. For each receive AP $r$, let
\[
\mathcal D_r
=
\left\{
(i,j):
1\le i<j\le J,\;
\bm\Sigma_i(\{r\})
\ne
\bm\Sigma_j(\{r\})
\right\}
\]
denote the set of symbol pairs distinguished by its local observation. Then $\cC$ is identifiable under $\cS$ if and only if
\begin{equation*}
\bigcup_{r\in\cS}\mathcal D_r
=
\{(i,j):1\le i<j\le J\}.
\end{equation*}
\end{cor}

\emph{Proof:} See the supplementary material.

Corollary~\ref{cor:independent_identifiability} gives a pair coverage interpretation of identifiability under independent receive AP observations. Each receive AP distinguishes a subset of symbol pairs, and the selected APs identify the embodied alphabet exactly when every pair is distinguished by at least one of them. This simple coverage characterization generally no longer holds when receive AP observations are correlated, because the contribution of an added observation can depend on its conditional relation to the observations already available.

Taken together, $\Pi_{\cS}$ records which symbol distinctions are absent under the selected observations, while $C_{ij}^{\star}(\cS)$ quantifies the strength of distinctions that are present. Receive AP cooperation can alter both by providing additional pairwise distinguishability beyond the observations already at the CPU. The next section fixes the receive AP set and characterizes how these statistical distinctions determine communication reliability over repeated sensing intervals.

\section{Reliability of the Embodied Alphabet}
\label{sec:reliability}

Fix a receive AP set $\cS$ and consider the transmission of a message $W$ through an embodied symbol $\bm q_W$ that is held fixed over $L$ sensing intervals. By the conditional independence model in Section~\ref{sec:model}, the interval observations are i.i.d. conditioned on $W$. Let $\bm y^{(\ell)}=\bm y_{\cS}^{(\ell)}$ denote the observation in interval $\ell$, define $\bm Y^{(L)}=(\bm y^{(1)},\ldots,\bm y^{(L)})$, and let $p_j^{(L)}$ denote its density under $W=j$, with dependence on $\cS$ suppressed. Then $p_j^{(L)}(\bm Y^{(L)})=\prod_{\ell=1}^{L}p_j^{\cS}(\bm y^{(\ell)})$. Hence, $p_i^{(L)}=p_j^{(L)}$ if and only if $p_i^{\cS}=p_j^{\cS}$, so repetition preserves the equivalence partition $\Pi_{\cS}$, while the pairwise Chernoff function and Chernoff information scale to $L C_{ij}(s;\cS)$ and $L C_{ij}^{\star}(\cS)$, respectively.

We use the maximum error probability so that reliability cannot be improved by sacrificing a particularly difficult embodied symbol. For a decoder $\delta$, define
\[
P_{\max}(\delta;L,\cC,\cS)
=
\max_{1\le j\le J}
\Pr\{\delta(\bm Y^{(L)})\ne j\mid W=j\}.
\]
The optimal maximum error probability is
\[
P_{\max}^{\star}(L;\cC,\cS)
=
\inf_{\delta}
P_{\max}(\delta;L,\cC,\cS).
\]
Randomized decoders are allowed, which permits the decision probabilities to be balanced among statistically indistinguishable symbols. The maximum likelihood (ML) decoder is
\[
\delta_{\rm ML}(\bm Y^{(L)})
=
\argmax_{1\le j\le J}
p_j^{(L)}(\bm Y^{(L)}),
\]
and its maximum error probability is denoted by $P_{\max}^{\rm ML}(L;\cC,\cS)$. Hence, $P_{\max}^{\star}(L;\cC,\cS)\le P_{\max}^{\rm ML}(L;\cC,\cS)$.

\subsection{Asymptotic Reliability}
\label{subsec:asymptotic_reliability}

The pairwise relations above do not yet determine the maximum error probability of the entire embodied alphabet. Define the largest equivalence class size as $\kappa_{\max}(\cC,\cS)=\max_{\mathcal E\in\Pi_{\cS}}|\mathcal E|$. An equivalence class contains physically different embodied symbols that induce exactly the same observation distribution under $\cS$, so $\kappa_{\max}$ measures the largest unresolved multiplicity in the alphabet.

For $J\ge2$, define the minimum pairwise Chernoff information as $C_{\min}^{\star}(\cC,\cS)=\min_{1\le i<j\le J}C_{ij}^{\star}(\cS)$. The alphabet is identifiable if and only if $\kappa_{\max}(\cC,\cS)=1$, equivalently $C_{\min}^{\star}(\cC,\cS)>0$. Once this condition holds, the pair or pairs attaining $C_{\min}^{\star}$ are the bottleneck pairs of the embodied alphabet.

Classical Chernoff theory characterizes the asymptotic discrimination of distinct hypotheses under independent repetitions \cite{chernoff1952measure,poor1994introduction}. The resulting reliability depends first on whether the selected receive APs leave any embodied symbols statistically indistinguishable. Repeated sensing then exhibits two different asymptotic behaviors, while receive AP expansion can change the regime itself. Fig.~\ref{fig:reliability_schematic} illustrates these reliability regimes and the effect of receive AP expansion. The following theorem characterizes these effects.

\begin{thm}[Reliability Characterization under Distributed Sensing]
\label{thm:reliability_characterization}
Consider a finite embodied alphabet $\cC$ with $J\ge2$.

\emph{1) Nonidentifiable regime:}
If $\kappa_{\max}(\cC,\cS)>1$, then
\begin{equation}
\lim_{L\to\infty}
P_{\max}^{\star}(L;\cC,\cS)
=
1-
\frac{1}{\kappa_{\max}(\cC,\cS)}.
\label{eq:ambiguity_error_floor}
\end{equation}

\emph{2) Identifiable regime:}
If $\kappa_{\max}(\cC,\cS)=1$, then
\begin{equation}
\begin{aligned}
\lim_{L\to\infty}
-\frac{1}{L}
\log P_{\max}^{\star}(L;\cC,\cS)
&=
C_{\min}^{\star}(\cC,\cS),\\
\lim_{L\to\infty}
-\frac{1}{L}
\log P_{\max}^{\rm ML}(L;\cC,\cS)
&=
C_{\min}^{\star}(\cC,\cS).
\end{aligned}
\label{eq:maximum_error_exponents}
\end{equation}

\emph{3) Receive AP expansion:}
For $\emptyset\ne\cS_1\subset\cS_2$,
\begin{equation}
\begin{aligned}
\kappa_{\max}(\cC,\cS_2)
&\le
\kappa_{\max}(\cC,\cS_1),\\
C_{\min}^{\star}(\cC,\cS_2)
&\ge
C_{\min}^{\star}(\cC,\cS_1).
\end{aligned}
\label{eq:reliability_monotonicity}
\end{equation}
The resulting reliability improvement has the following three cases.

\emph{(a) Error floor reduction:}
Suppose $\cC$ is nonidentifiable under both $\cS_1$ and $\cS_2$. The exact asymptotic error floor strictly decreases if and only if $\kappa_{\max}(\cC,\cS_2)<\kappa_{\max}(\cC,\cS_1)$.

\emph{(b) Error floor elimination:}
If $\kappa_{\max}(\cC,\cS_1)>1$ and $\kappa_{\max}(\cC,\cS_2)=1$, then the asymptotic error floor is eliminated, and the error under $\cS_2$ decays exponentially with exponent $C_{\min}^{\star}(\cC,\cS_2)>0$. By Theorem~\ref{thm:additional_ap_distinguishability}, this occurs if and only if every $i<j$ satisfying $i\sim_{\cS_1}j$ also satisfies $\bm A_i\ne\bm A_j$ or $\bm V_i\ne\bm V_j$.

\emph{(c) Error exponent improvement:}
If $\cC$ is already identifiable under $\cS_1$, then $C_{\min}^{\star}(\cC,\cS_2)>C_{\min}^{\star}(\cC,\cS_1)$ if and only if the additional observations provide a new distinction for every pair attaining the minimum pairwise Chernoff information under $\cS_1$. Specifically, every pair $(i,j)$ with $C_{ij}^{\star}(\cS_1)=C_{\min}^{\star}(\cC,\cS_1)$ must satisfy $\bm A_i\ne\bm A_j$ or $\bm V_i\ne\bm V_j$.
\end{thm}

\begin{figure}[t]
\centering
\includegraphics[width=0.74\linewidth]{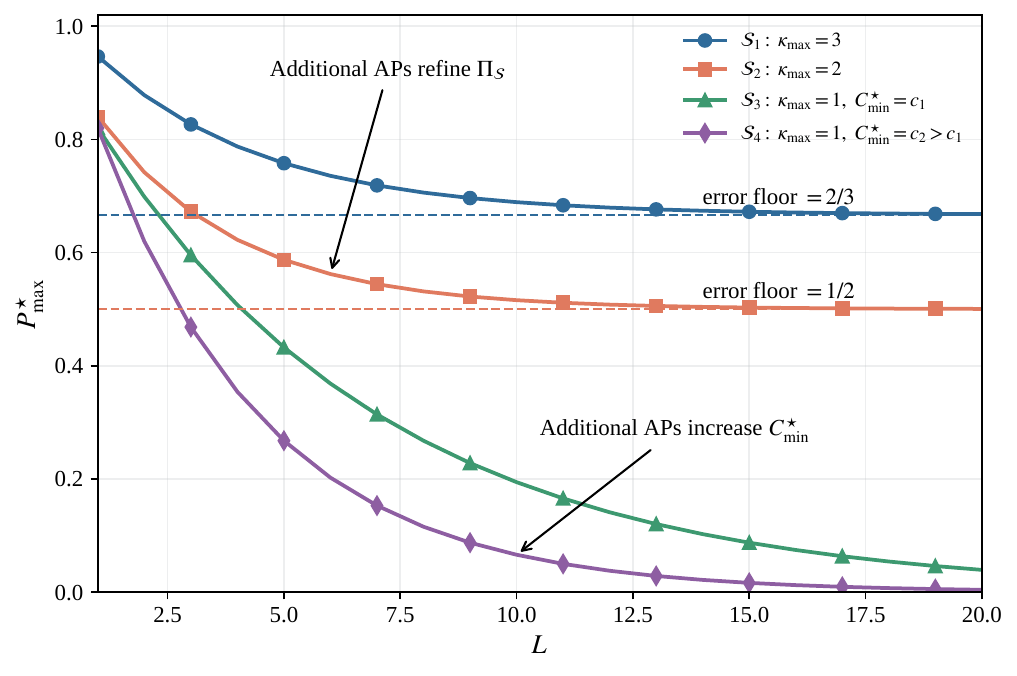}
\caption{Receive AP expansion. Reducing $\kappa_{\max}$ lowers the asymptotic error floor. Once $\kappa_{\max}=1$, increasing $C_{\min}^{\star}$ improves the error exponent.}
\label{fig:reliability_schematic}
\end{figure}

\emph{Proof:} See the supplementary material.

Theorem~\ref{thm:reliability_characterization} reveals fundamentally different roles of spatial and temporal sensing resources. Temporal repetition accumulates evidence for statistical distinctions already available under a fixed receive AP set, but cannot create a missing distinction. Receive AP cooperation changes the available observations and can therefore both resolve statistical ambiguities and strengthen existing distinctions. In particular, an embodied alphabet may be nonidentifiable under every individual receive AP but identifiable under their joint observation. In this regime, receive AP cooperation removes an error floor that no individual receive AP can overcome through temporal repetition.

\subsection{Reliability with a Finite Sensing Budget}
\label{subsec:finite_sensing_budget}

For a prescribed finite number of sensing intervals $L$ and target maximum error probability $\epsilon$, the exact maximum error probability is generally unavailable. We therefore establish a sufficient condition under ML decoding and a necessary condition that applies to any decoder, with $\epsilon\in(0,1/2)$.

\begin{prop}[Sufficient Reliability Condition]
\label{prop:maximum_error_bound}
The ML decoder satisfies
\begin{equation}
P_{\max}^{\rm ML}(L;\cC,\cS)
\le
(J-1)
\exp
\left[
-LC_{\min}^{\star}(\cC,\cS)
\right].
\label{eq:maximum_error_bound}
\end{equation}
Consequently,
\begin{equation}
C_{\min}^{\star}(\cC,\cS)
\ge
C_{\rm req}(J,\epsilon,L)
\triangleq
\frac{1}{L}
\log\frac{J-1}{\epsilon}
\label{eq:reliability_requirement}
\end{equation}
is sufficient for $P_{\max}^{\rm ML}(L;\cC,\cS)\le\epsilon$.
\end{prop}

\emph{Proof:} See the supplementary material.

The sufficient condition identifies embodied alphabets that are guaranteed to meet the target reliability under ML decoding within the available sensing budget. A complementary necessary condition identifies a region that no decoder can support.

\begin{prop}[Necessary Reliability Condition]
\label{prop:pairwise_necessary_condition}
Suppose a decoder has maximum error probability no greater than $\epsilon<1/2$. Then
\begin{equation}
C_{\min}^{\star}(\cC,\cS)
\ge
C_{\rm nec}(\epsilon,L)
\triangleq
\frac{1}{2L}
\log
\frac{1}{4\epsilon(1-\epsilon)}.
\label{eq:necessary_reliability_requirement}
\end{equation}
\end{prop}

\emph{Proof:} See the supplementary material.

The two bounds divide the finite sensing problem into three regions. If $C_{\min}^{\star}<C_{\rm nec}$, the target reliability is impossible for any decoder. If $C_{\min}^{\star}\ge C_{\rm req}$, it is guaranteed under ML decoding. The interval between the two thresholds is not resolved by these bounds.

The thresholds also reveal a finite sensing benefit of receive AP cooperation. If $C_{\min}^{\star}(\cC,\{r\})<C_{\rm nec}(\epsilon,L)$ for every $r\in\cS$, while $C_{\min}^{\star}(\cC,\cS)\ge C_{\rm req}(J,\epsilon,L)$, then no decoder using any individual receive AP can meet the target maximum error probability $\epsilon$, whereas the joint ML decoder over $\cS$ is guaranteed to meet it. Thus, for the same embodied alphabet and sensing budget, joint processing across the selected receive APs can move the system from a region where the target reliability is impossible to one where it is guaranteed under ML decoding.

\subsection{Communication Rate}
\label{subsec:communication_rate}

For a fixed design $(\cC,\cS)$ with $|\cC|=J$, Proposition~\ref{prop:maximum_error_bound} guarantees the target maximum error probability whenever $C_{\min}^{\star}(\cC,\cS)\ge C_{\rm req}(J,\epsilon,L)$. For any design satisfying this condition, one embodied symbol conveys $\log_2J$ bits over $L$ sensing intervals. The corresponding communication rate over the sensing duration is
\begin{equation}
R_{\epsilon,L}(\cC,\cS)
=
\frac{\log_2 J}{LT_{\rm sen}}.
\label{eq:embodied_rate}
\end{equation}

The normalization in \eqref{eq:embodied_rate} includes only the sensing duration and excludes actuation, settling, and waiting time. It therefore characterizes communication during the sensing operation rather than the complete end-to-end throughput.

Increasing $L$ relaxes the pairwise Chernoff information required for a prescribed reliability target, but also increases the sensing duration in \eqref{eq:embodied_rate}. Additional receive APs can instead improve the reliability condition without increasing this duration. The resulting design problem is to select the receive APs and embodied alphabet so as to support as many reliable embodied symbols as possible under the prescribed sensing budget, which is addressed in Section~\ref{sec:design}.

\section{Joint Receive AP and Embodied Alphabet Design}
\label{sec:design}

The reliability analysis in Section~\ref{sec:reliability} reveals a coupling between the receive AP set and the embodied alphabet. The embodied alphabet $\cC$ determines the statistical distinctions required for communication, while the receive AP set $\cS$ determines the observations available at the CPU. An AP that is useful for one alphabet can be redundant for another, and changing $\cS$ can alter both the unresolved equivalence classes and the symbol pairs that limit reliability. The two objects must therefore be designed jointly.

Related sensor, AP, and receiver selection problems have been studied in distributed detection and ISAC \cite{bajovic2011sensor,yan2024communicate,wang2026receiver}. Here, the hypotheses themselves are determined by the selected scatterer positions rather than being fixed independently of the receive AP selection. Consequently, a receive AP cannot be assigned a standalone sensing quality independently of $\cC$. The design must account for both whether the required statistical distinctions exist and how strong they are.

\subsection{Joint Design Criterion}
\label{subsec:joint_design}

Let $\cQ_{\rm G}\subseteq\cQ$ be a finite candidate position set with $N_{\rm G}=|\cQ_{\rm G}|$, and let $N_{\rm rx}\in\{1,\ldots,|\cR|\}$ be the receive AP budget.
We first solve the design for a prescribed alphabet size $J$, with $2\le J\le N_{\rm G}$, and later repeat the design over candidate values of $J$.

The reliability characterization in Section~\ref{sec:reliability} identifies two quantities with direct asymptotic meanings for the joint design:
\begin{itemize}[leftmargin=0.45cm]
\item \emph{Ambiguity resolution:} for a nonidentifiable alphabet, $\kappa_{\max}(\cC,\cS)$ determines the asymptotic error floor.
\item \emph{Distinction strengthening:} once the alphabet is identifiable, $C_{\min}^{\star}(\cC,\cS)$ determines the error exponent.
\end{itemize}
Direct maximization of $C_{\min}^{\star}$ from the outset cannot distinguish among nonidentifiable designs because every such design has $C_{\min}^{\star}=0$, even though their unresolved ambiguity structures can be very different.

To describe this structure, arrange the equivalence class sizes in descending order and pad the vector with zeros to length $J$:
\begin{equation*}
\bm\kappa(\cC,\cS)
=
\operatorname{sort}_{\downarrow}
\bigl(
(|\mathcal E|)_{\mathcal E\in\Pi_{\cS}},
0,\ldots,0
\bigr).
\end{equation*}
Its first entry is $\kappa_{\max}(\cC,\cS)$ and therefore determines the exact asymptotic error floor whenever the alphabet is nonidentifiable.
If two such designs have the same largest equivalence class, the asymptotic maximum error alone does not distinguish them.
We use the remaining class sizes to break this tie, preferring smaller unresolved multiplicities.

To compare the remaining statistical distinctions after the ambiguity structure has been fixed, arrange all pairwise Chernoff information values in ascending order:
\begin{equation*}
\bm c(\cC,\cS)
=
\operatorname{sort}_{\uparrow}
\bigl(
(C_{ij}^{\star}(\cS))_{1\le i<j\le J}
\bigr).
\end{equation*}
Its first entry is $C_{\min}^{\star}(\cC,\cS)$.
When the alphabet is identifiable, all equivalence classes are singletons and hence $\bm\kappa=(1,\ldots,1)$.
The comparison then passes to $\bm c$, whose first entry determines the asymptotic error exponent.
If the weakest pairwise distinction is identical for two designs, the subsequent entries break the tie by successively comparing their next weakest pairs.

We compare feasible designs lexicographically: a smaller $\bm\kappa$ is preferred, and if $\bm\kappa$ is identical, a larger $\bm c$ is preferred.
The resulting joint design problem is
\begin{equation}
\begin{aligned}
\underset{\cC,\cS}{\operatorname{lex\,minimize}}\quad
&\bigl(
\bm\kappa(\cC,\cS),
-\bm c(\cC,\cS)
\bigr)\\
\mathrm{s.t.}\quad
&\cC\subseteq\cQ_{\rm G},
\qquad
|\cC|=J,\\
&\cS\subseteq\cR,
\qquad
|\cS|=N_{\rm rx}.
\end{aligned}
\label{eq:joint_ap_alphabet_design}
\end{equation}

The equality constraint $|\cS|=N_{\rm rx}$ does not restrict the optimum relative to a budget constraint $|\cS|\le N_{\rm rx}$.
By Theorem~\ref{thm:additional_ap_distinguishability}, adding receive AP observations can only refine $\Pi_{\cS}$ and cannot decrease any pairwise Chernoff information.
Hence, any design using fewer than $N_{\rm rx}$ receive APs can be enlarged without worsening the lexicographic criterion.

\subsection{Receive AP Selection Structure}
\label{subsec:receive_ap_selection}

The receive AP expansion in Theorem~\ref{thm:reliability_characterization} shows that the value of an additional receive AP depends on the distinctions already available at the CPU. In the nonidentifiable regime, an AP can be useful by splitting an unresolved equivalence class. In the identifiable regime, $C_{\min}^{\star}$ increases only if every pair attaining the current minimum gains additional pairwise distinguishability. A receive AP should therefore not be ranked by a standalone measure of sensing strength.

To expose the combinatorial structure behind receive AP selection, we temporarily fix the embodied alphabet $\cC$ and focus on the minimum pairwise Chernoff information. For nonempty $\cS$, define $f_{\cC}(\cS)=C_{\min}^{\star}(\cC,\cS)$ and set $f_{\cC}(\emptyset)=0$. This scalar objective does not replace the lexicographic criterion in \eqref{eq:joint_ap_alphabet_design}; it is used here to determine whether receive AP gains satisfy the diminishing returns property underlying standard greedy guarantees. The following proposition shows that they generally do not.

\begin{prop}[Monotonicity and Failure of Submodularity]
\label{prop:ap_objective_structure}
For a fixed $\cC$, $f_{\cC}(\cS)$ is monotone in $\cS$, but it is not generally submodular.
\end{prop}

\emph{Proof:} See the supplementary material.

Proposition~\ref{prop:ap_objective_structure} formalizes receive AP complementarity. In the counterexample, either AP alone gives zero gain in the minimum pairwise Chernoff information, whereas their joint observations identify the alphabet and yield a positive minimum. More generally, the marginal value of a receive AP can increase after other APs have been selected because the ambiguity and bottleneck structure changes with $\cS$.

This complementarity rules out the diminishing returns structure required by the standard approximation guarantee for greedy maximization of a monotone submodular function \cite{nemhauser1978analysis}. A one-way AP insertion procedure can therefore commit to an unfavorable early choice when later APs change the active ambiguities or bottleneck pairs. Since these changes can also alter the preferred embodied alphabet, the joint search developed below allows both AP exchanges and alphabet refinement.

\subsection{Efficient Chernoff Evaluation}
\label{subsec:efficient_chernoff_evaluation}

The joint design repeatedly evaluates pairwise Chernoff information for many candidate position pairs and receive AP sets. Lemma~\ref{lem:gaussian_chernoff} reduces the optimization over the Chernoff parameter to a scalar root search, but each evaluation of $C_{ij}(s)$ still involves determinants in the full CPU observation dimension $D_{\cS}$. From \eqref{eq:block_cov}, the embodied symbol affects only the echo covariance, whereas $\bm K(\cS)$ is common to all hypotheses. Since $\bm K(\cS)\succ\bm0$, whitening with $\bm K^{-1/2}(\cS)$ maps this common residual and noise covariance to the identity. The statistical differences among embodied symbols then remain only in the whitened echo covariance.

For a fixed $\cS$, suppress its dependence and write $\bm K=\bm K(\cS)$, $\bm G_j=\bm G_{\cS,j}$, $\bm R_{\alpha,j}=\bm R_{\alpha,\cS,j}$, and $\bm\Sigma_j=\bm\Sigma_j(\cS)$. Let $\bm U_j$ be any full column rank factor satisfying
\[
\bm K^{-1/2}
\bm G_j
\bm R_{\alpha,j}
\bm G_j^H
\bm K^{-1/2}
=
\bm U_j\bm U_j^H.
\]

\begin{prop}[Reduced Chernoff Representation]
\label{prop:signal_subspace_reduction}
For $s\in[0,1]$, define $\bm U_{ij}(s)=[\sqrt{1-s}\bm U_i,\sqrt{s}\bm U_j]$. Then
\begin{equation}
\begin{aligned}
C_{ij}(s)
&=
\log\det
\bigl(
\bm I+
\bm U_{ij}^H(s)
\bm U_{ij}(s)
\bigr)\\
&\quad
-(1-s)
\log\det
\bigl(
\bm I+\bm U_i^H\bm U_i
\bigr)\\
&\quad
-s
\log\det
\bigl(
\bm I+\bm U_j^H\bm U_j
\bigr).
\end{aligned}
\label{eq:reduced_chernoff}
\end{equation}
Hence, the determinants in \eqref{eq:reduced_chernoff} have order at most $\rank(\bm U_i)+\rank(\bm U_j)$, rather than $D_{\cS}$.
\end{prop}

\emph{Proof:} See the supplementary material.

The reduction also reveals which observation dimensions can contribute to pairwise distinguishability. In the whitened coordinates, the two hypotheses can differ only on $\operatorname{span}([\bm U_i,\bm U_j])$. On its orthogonal complement, both covariance matrices are the identity, so these directions contribute no Chernoff information. Equivalently, the associated generalized eigenvalues in Lemma~\ref{lem:gaussian_chernoff} are equal to one. Consequently, no more than $\rank([\bm U_i,\bm U_j])\le\rank(\bm U_i)+\rank(\bm U_j)$ generalized eigenvalues can differ from one.

In the present model, $\rank(\bm U_j)\le\rank(\bm G_j)\le|\cP||\cS|$. The determinant dimension is therefore governed by the echo subspace rank rather than the full observation dimension $D_{\cS}$. This reduction is particularly useful in the joint search below, where Chernoff information must be evaluated for many candidate position pairs and receive AP sets.

\subsection{Joint Search Algorithm}
\label{subsec:joint_search}

Receive AP complementarity motivates revisiting earlier AP choices through exchange, while changes in $\cS$ can change the alphabet preferred by the lexicographic criterion. We therefore refine the alphabet whenever $\cS$ changes and retain multiple starts to reduce sensitivity to early AP choices. Proposition~\ref{prop:signal_subspace_reduction} reduces the cost of the repeated Chernoff evaluations.

For a fixed receive AP set $\cS$, physical separation alone does not reliably reflect statistical distinguishability at the CPU. We initialize the alphabet with the candidate position pair having the largest Chernoff information,
\[
(\bm q_1,\bm q_2)
\in
\argmax_{\substack{
\bm q,\widetilde{\bm q}\in\cQ_{\rm G}\\
\bm q\ne\widetilde{\bm q}
}}
C^{\star}(\bm q,\widetilde{\bm q};\cS),
\]
where $C^{\star}(\bm q,\widetilde{\bm q};\cS)$ is the Chernoff information induced by the two candidate positions under $\cS$. Starting from this pair, positions are inserted until the alphabet contains $J$ symbols. At each step, every unused position is tentatively added and the resulting designs are compared by \eqref{eq:joint_ap_alphabet_design}. Because later insertions can change the ambiguity structure and bottleneck pairs, the completed alphabet is refined by single position replacements until no strict improvement remains; ties may be broken arbitrarily.

The receive AP search starts from singleton sets. For each $r\in\cR$, an alphabet is constructed and refined under $\cS=\{r\}$, and up to $N_{\rm init}$ preferred pairs are retained. From each start, APs are inserted until $N_{\rm rx}$ is reached. For every candidate addition, the alphabet is refined under the enlarged receive AP set before the joint designs are compared. The insertion continues even when the immediate increase in $C_{\min}^{\star}$ is zero, since an AP can improve the equivalence structure before identifiability and complementary APs can yield a later gain.

After the budget is filled, exchanges between selected and unselected APs are tested, again refining the alphabet before each comparison. Strictly preferred exchanges are accepted until no single exchange improves the design. Algorithm~\ref{alg:joint_design} summarizes the search.

\begin{algorithm}[t]
\caption{Joint receive AP and embodied alphabet search}
\label{alg:joint_design}
\begin{algorithmic}[1]
\Require $\cQ_{\rm G}$, $\cR$, $J$, $N_{\rm rx}$, and $N_{\rm init}$
\Ensure $(\widehat{\cC}_J,\widehat{\cS}_J)$
\For{$r\in\cR$}
    \State Set $\cS_r\gets\{r\}$; construct and refine $\cC_r$ with $J$ symbols under $\cS_r$
\EndFor
\State Retain the $\min\{N_{\rm init},|\cR|\}$ most preferred pairs $(\cC_r,\cS_r)$
\For{each retained pair $(\cC,\cS)$}
    \While{$|\cS|<N_{\rm rx}$}
        \For{$r^{+}\in\cR\setminus\cS$}
            \State Refine a copy of $\cC$ under $\cS\cup\{r^{+}\}$ to obtain a candidate
        \EndFor
        \State Update $(\cC,\cS)$ to the most preferred candidate
    \EndWhile
    \If{$|\cS|<|\cR|$}
        \State $\mathrm{improved}\gets\mathrm{true}$
        \While{$\mathrm{improved}$}
            \For{each $(r^{-},r^{+})\in\cS\times(\cR\setminus\cS)$}
                \State Refine a copy of $\cC$ under $(\cS\setminus\{r^{-}\})\cup\{r^{+}\}$ to obtain a candidate
            \EndFor
            \State Let $(\cC',\cS')$ be the most preferred candidate
            \If{$(\cC',\cS')$ is preferred to $(\cC,\cS)$}
                \State $(\cC,\cS)\gets(\cC',\cS')$
            \Else
                \State $\mathrm{improved}\gets\mathrm{false}$
            \EndIf
        \EndWhile
    \EndIf
    \State Store the final pair from this start
\EndFor
\State Return the most preferred stored pair as $(\widehat{\cC}_J,\widehat{\cS}_J)$
\end{algorithmic}
\end{algorithm}

Since $\cQ_{\rm G}$ and $\cR$ are finite, the insertion stages contain finitely many updates. The position refinement and AP exchange stages also terminate because every accepted update strictly improves the lexicographic criterion over a finite set of feasible designs.

\textit{Computational complexity.}
The dominant operation is pairwise Chernoff evaluation. For each receive AP set, the $N_{\rm G}(N_{\rm G}-1)/2$ candidate position pairs can be evaluated once and cached; Lemma~\ref{lem:gaussian_chernoff} reduces each optimization to a scalar root search, and Proposition~\ref{prop:signal_subspace_reduction} reduces the determinant dimension. After the initial pair is selected, alphabet construction examines $\sum_{k=2}^{J-1}(N_{\rm G}-k)$ candidate insertions, while one position refinement sweep examines $J(N_{\rm G}-J)$ replacements. With $n<N_{\rm rx}$ selected APs, one insertion step considers $|\cR|-n$ enlarged sets, and one exchange sweep at the budget considers $N_{\rm rx}(|\cR|-N_{\rm rx})$ exchanges. Each candidate AP change invokes alphabet refinement, and for fixed $N_{\rm G}$, $|\cR|$, $J$, and $N_{\rm rx}$, the effort grows linearly with $N_{\rm init}$.

\subsection{Alphabet Size Selection}
\label{subsec:alphabet_size_selection}

Algorithm~\ref{alg:joint_design} optimizes the receive AP set and embodied alphabet for a prescribed alphabet size $J$. To determine the alphabet size, we repeat the joint search for $J=2,\ldots,N_{\rm G}$ and test each returned design against the reliability condition in Proposition~\ref{prop:maximum_error_bound}. Let $\widehat F_J=C_{\min}^{\star}(\widehat{\cC}_J,\widehat{\cS}_J)$ denote the minimum pairwise Chernoff information of the design returned for size $J$. For prescribed $L$ and target maximum error probability $\epsilon$, this design is guaranteed to satisfy the reliability requirement whenever $\widehat F_J\ge C_{\rm req}(J,\epsilon,L)$.

Since $L$ and $T_{\rm sen}$ are fixed, the communication rate $\log_2 J/(LT_{\rm sen})$ increases monotonically with $J$. We therefore select the largest alphabet size among the returned designs satisfying the reliability condition:
\begin{equation*}
\widehat J_{\epsilon,L}
=
\max
\left(
\{1\}
\cup
\left\{
J:
\substack{
2\le J\le N_{\rm G},\\
\widehat F_J\ge C_{\rm req}(J,\epsilon,L)
}
\right\}
\right).
\end{equation*}
The value $\widehat J_{\epsilon,L}=1$ applies when none of the tested nontrivial alphabets satisfies the reliability requirement. The resulting communication rate over the $L$ sensing intervals is $\widehat R_{\epsilon,L}=\log_2\widehat J_{\epsilon,L}/(LT_{\rm sen})$.

\section{Numerical Results}
\label{sec:numerical_results}

Unless stated otherwise, the controllable region is $9\,\mathrm{m}\times9\,\mathrm{m}$ at height $h_{\rm A}=1.2$ m. The four transmit APs are located at $(-38,-30,10)$, $(36,-28,9)$, $(32,38,11)$, and $(-40,32,8.5)$ m, and every AP uses a $2\times2$ UPA whose broadside points toward the region center. Each sensing interval contains $\tau_s=6$ probing symbols. The fixed probing block uses unit modulus symbols and steers each transmit AP toward the region center.

We use the common local scattering covariance $\bm R_{\alpha}=(1-\rho_{\alpha})\bm I_{|\cP|}+\rho_{\alpha}\bm1_{|\cP|}\bm1_{|\cP|}^T$ with $\rho_{\alpha}=0.2$, so $\bm R_{\alpha,r,j}=\bm R_{\alpha}$ for all $r$ and $j$. Scattering coefficients are independent across receive APs except in the experiment with correlated observations below, while residual terms remain independent across receive APs. The local residual and noise covariance is $\bm K_r=(1+\xi)\bm I_D+\xi D\bm v_r\bm v_r^H$, where $\|\bm v_r\|=1$ and $\xi=0.08$. The vector $\bm v_r$ is the normalized response across the array and probing symbols of a fixed residual scatterer at $(10,-8,h_{\rm A})$ m. Following the bistatic range law in \cite{behdad2024multistatic}, we use $\beta_{r,t,j}\propto\bigl(\|\bm q_j-\bm b_t\|\|\bm q_j-\bm b_r\|/(\|[0,0,h_{\rm A}]^T-\bm b_t\|\|[0,0,h_{\rm A}]^T-\bm b_r\|)\bigr)^{-2}$. The common echo scale is calibrated so that $\mathrm{SNR}_{\rm ref}\triangleq\tr(\bm G_{r,j}\bm R_{\alpha}\bm G_{r,j}^H)/\tr(\bm K_r)$ at $\bm q_j=[0,0,h_{\rm A}]^T$ with the normalized bistatic factors set to one.

\subsection{Pairwise Distinguishability}
\label{subsec:num_pairwise_distinguishability}

We first evaluate the gain from optimizing the Chernoff parameter and then examine the incremental value of an additional receive AP.

\begin{figure}[!t]
\centering
\includegraphics[width=0.99\linewidth]{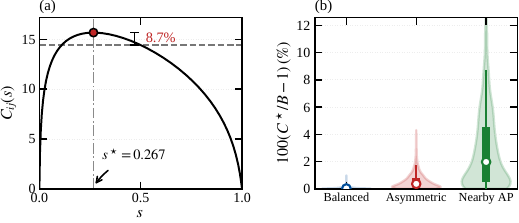}
\caption{Chernoff distinguishability of two embodied symbols. (a) Chernoff function for a representative pair with a nearby receive AP. (b) Relative gain over the Bhattacharyya distance under increasingly asymmetric sensing layouts.}
\label{fig:chernoff_bd}
\end{figure}

\begin{figure}[!t]
\centering
\includegraphics[width=0.6\linewidth]{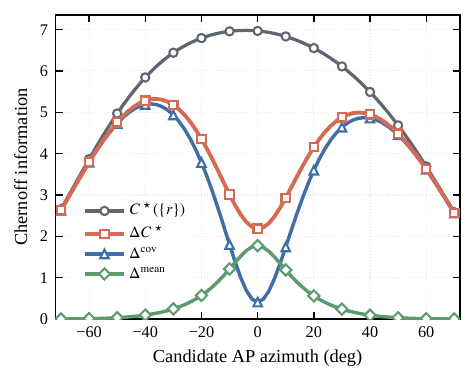}
\caption{Additional receive AP contribution versus candidate AP azimuth: standalone Chernoff information, optimized increment, and conditional covariance and mean terms evaluated at the maximizing Chernoff parameter of the enlarged set.}
\label{fig:additional_ap_contribution}
\end{figure}

Fig.~\ref{fig:chernoff_bd} evaluates the effect of optimizing the Chernoff parameter with two receive APs and $\mathrm{SNR}_{\rm ref}=12$ dB. Let $d_r$ and $h_r$ denote the horizontal distance and height of receive AP $r$, respectively. We consider three layouts with progressively stronger geometric imbalance:
\begin{itemize}[leftmargin=0.5cm]
    \item \emph{Balanced:} $d_1,d_2\in[25,40]$ m.
    \item \emph{Asymmetric:} $d_1\in[8,13]$ m and $d_2\in[35,60]$ m.
    \item \emph{Nearby AP:} $d_1\in[4.2,5]$ m, $h_1\in[2,4]$ m, and $d_2\in[35,60]$ m.
\end{itemize}
For the representative symbol pair in Fig.~\ref{fig:chernoff_bd}(a), $s^{\star}=0.267$ and $C_{ij}^{\star}$ is $8.7\%$ larger than $B_{ij}$. Across random scatterer position pairs, the median relative gains are $0.03\%$, $0.37\%$, and $1.98\%$ for the three layouts, respectively, and the $95$th percentile reaches $8.63\%$ in the nearby AP case.

In balanced layouts, the Chernoff function remains nearly symmetric about $s=1/2$, so $B_{ij}$ closely approximates $C_{ij}^{\star}$. Geometric imbalance can shift the maximizer from the midpoint and create an appreciable gap for some symbol pairs.

We next examine whether the receive AP that is strongest in isolation is also the most useful addition to an existing receive AP set. Let $r_0$ denote the fixed receive AP at $(28,0,8)$ m and $\cS_0=\{r_0\}$, while candidate receive AP $r$ moves along a horizontal arc of radius $30$ m at height $8$ m. The two scatterer positions are $(0,-3,h_{\rm A})$ and $(0,3,h_{\rm A})$ m, and $\mathrm{SNR}_{\rm ref}=16$ dB. The bistatic gain uses the normalization $\beta_{r,t,j}\propto\bigl(\|\bm q_j-\bm b_t\|\|\bm q_j-\bm b_r\|/(45\,\mathrm{m}\times35\,\mathrm{m})\bigr)^{-2}$. If $\phi$ denotes the azimuth separation between the two receive APs, cross-AP scattering correlation is modeled by $\rho_{\rm rx}(\phi)=0.99\exp[-(|\phi|/45^\circ)^2]$ and $\bm R_{\alpha,\cS_0\cup\{r\},j}=\bm R_{\rm rx}(\phi)\otimes\bm R_{\alpha}$, where $\bm R_{\rm rx}(\phi)=[1-\rho_{\rm rx}(\phi)]\bm I_2+\rho_{\rm rx}(\phi)\bm1_2\bm1_2^T$.

\begin{figure}[!t]
\centering
\includegraphics[width=0.76\linewidth]{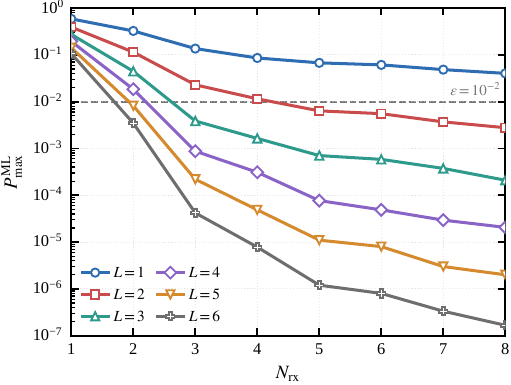}
\caption{Estimated maximum ML decoding error versus $N_{\rm rx}$ for different numbers $L$ of sensing intervals, with a fixed embodied alphabet.}
\label{fig:reliability_transition}
\end{figure}

Define $\Delta C_{ij}^{\star}(r)=C_{ij}^{\star}(\cS_0\cup\{r\})-C_{ij}^{\star}(\cS_0)$. Fig.~\ref{fig:additional_ap_contribution} compares the standalone value $C_{ij}^{\star}(\{r\})$ with $\Delta C_{ij}^{\star}(r)$. The two rankings differ. At $-4^\circ$, the candidate AP has the largest standalone value, $6.979$, with an increment of $2.343$. At $-36^\circ$, its standalone value decreases to $6.112$, and the increment reaches its maximum of $5.323$. Standalone pairwise distinguishability therefore does not determine the value of an AP once the observations over $\cS_0$ are available, consistent with Theorem~\ref{thm:additional_ap_distinguishability}.

At $s_r^{\star}\in\argmax_s C_{ij}(s;\cS_0\cup\{r\})$, the increment at the fixed Chernoff parameter is dominated by the mean term at $-4^\circ$, with $\Delta_{ij}^{\rm mean}=1.656$ and $\Delta_{ij}^{\rm cov}=0.687$, whereas the covariance term dominates at $-36^\circ$, with $5.184$ versus $0.143$. As characterized after Theorem~\ref{thm:additional_ap_distinguishability}, these terms decompose the Chernoff function increment at $s_r^{\star}$, while the optimized increment also accounts for optimization over $s$. Thus, an additional receive AP can provide new distinguishability through either the conditional prediction relation or the innovation covariance.

\subsection{Reliability with Networked Sensing}
\label{subsec:num_ap_selection}

We next fix $\cC$ and examine how receive AP cooperation and repeated sensing affect alphabet-level decoding reliability.

For a fixed alphabet with $J=6$, eight candidate receive APs, $\mathrm{SNR}_{\rm ref}=4$ dB, and $\epsilon=0.01$, we use the planar scatterer coordinates $(-4.5,-4.5)$, $(4.5,4.5)$, $(-4.5,4.5)$, $(4.5,-4.5)$, $(0,2.25)$, and $(-4.5,0)$ m at height $h_{\rm A}$. Exhaustive subset search selects $\cS$ to maximize $C_{\min}^{\star}(\cC,\cS)$ for each $N_{\rm rx}$. Fig.~\ref{fig:reliability_transition} reports the estimated maximum ML decoding error for $L=1,\ldots,6$. Over the tested receive AP budgets, the target is not reached for $L=1$. It is first reached at $N_{\rm rx}=5$ for $L=2$, at $N_{\rm rx}=3$ for $L=3$ and $4$, and at $N_{\rm rx}=2$ for $L=5$ and $6$.

\begin{figure}[!t]
\centering
\includegraphics[width=0.62\linewidth]{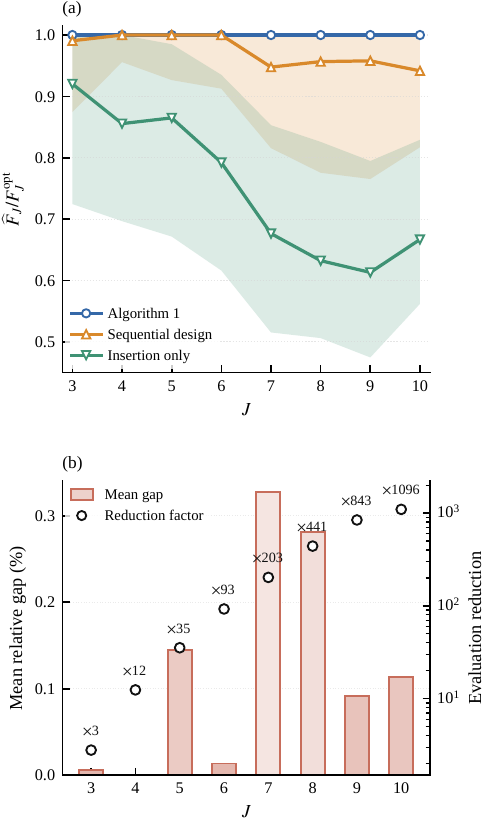}
\caption{Algorithm~\ref{alg:joint_design} and two reduced procedures against the exact discrete benchmark on reduced instances. (a) Ratio $\widehat F_J/F_J^{\rm opt}$; markers show medians and shaded bands span the $10$th--$90$th percentiles over $100$ layouts. (b) Mean relative gap and reduction in candidate evaluations relative to exhaustive enumeration.}
\label{fig:joint_design_validation}
\end{figure}

\begin{figure*}[!t]
\centering
\includegraphics[width=0.78\textwidth]{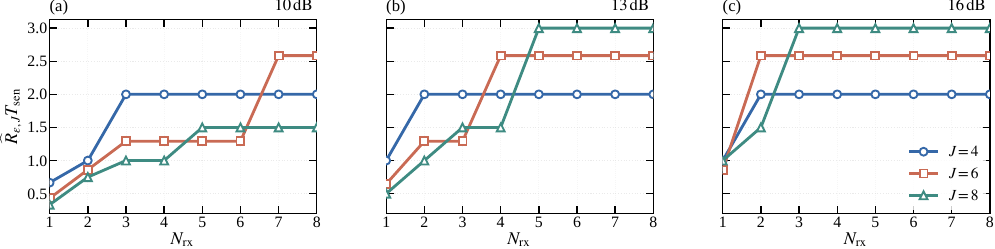}
\caption{Median guaranteed normalized embodied communication rate versus $N_{\rm rx}$ for jointly designed alphabets with $J\in\{4,6,8\}$ at $\mathrm{SNR}_{\rm ref}\in\{10,13,16\}$ dB over $100$ random receive AP layouts.}
\label{fig:embodied_rate_scaling}
\end{figure*}

\subsection{Joint Design Validation}
\label{subsec:num_joint_design_validation}

Exact validation uses a $5\times5$ candidate position grid, eight candidate receive APs, and $N_{\rm rx}=3$. For each $J$, $F_J^{\rm opt}$ is the largest $C_{\min}^{\star}$ over all alphabets and receive AP subsets. For each subset, thresholded clique feasibility gives the exact maximin Chernoff value without explicitly enumerating every alphabet. We set $\mathrm{SNR}_{\rm ref}=11$ dB, $J=3,\ldots,10$, use $100$ random receive AP layouts, and retain $N_{\rm init}=8$ starts.

The sequential design selects each new receive AP using the current $\cC$, refines $\cC$ only after selection, and omits the final AP exchange. The insertion only design uses receive AP and position insertion without subsequent position refinement or AP exchange.

Fig.~\ref{fig:joint_design_validation}(a) shows that the median $\widehat F_J/F_J^{\rm opt}$ of Algorithm~\ref{alg:joint_design} equals one for every tested $J$, and its mean relative gap never exceeds $0.33\%$. At $J=10$, the median ratios of the sequential and insertion only designs decrease to $0.942$ and $0.667$, respectively.

Fig.~\ref{fig:joint_design_validation}(b) shows that the reduction in candidate evaluations relative to explicit exhaustive enumeration increases from approximately $2.8\times$ to $1.10\times10^3$ as $J$ grows from $3$ to $10$.

\subsection{Guaranteed Communication Rate}
\label{subsec:num_guaranteed_rate}

For fixed $J$ and each receive AP budget, let $\widehat L_{\epsilon,J}$ be the smallest positive integer satisfying $\widehat F_J\ge C_{\rm req}(J,\epsilon,\widehat L_{\epsilon,J})$ when $\widehat F_J>0$. From \eqref{eq:embodied_rate}, the corresponding rate is $\widehat R_{\epsilon,J}=\log_2J/(\widehat L_{\epsilon,J}T_{\rm sen})$.

Using the same candidate position grid and receive AP pool as in the preceding experiment, we set $\epsilon=0.01$ and evaluate $J\in\{4,6,8\}$ with $\mathrm{SNR}_{\rm ref}\in\{10,13,16\}$ dB over $100$ random receive AP layouts. Fig.~\ref{fig:embodied_rate_scaling} plots the median normalized rate $\widehat R_{\epsilon,J}T_{\rm sen}$ versus $N_{\rm rx}$.

At $10$ dB, the ceiling corresponding to $\widehat L_{\epsilon,J}=1$ is reached at $N_{\rm rx}=3$ for $J=4$ and at $N_{\rm rx}=7$ for $J=6$. The $J=8$ design remains below its ceiling over the tested AP budgets. At $13$ dB, the $J=6$ and $J=8$ ceilings are reached at $N_{\rm rx}=4$ and $5$, respectively. At $16$ dB, these budgets decrease to $2$ and $3$.

The steps arise from the integer sensing duration $\widehat L_{\epsilon,J}$: additional receive APs may increase $\widehat F_J$ without changing the rate until the required number of sensing intervals decreases. The rate saturates when $\widehat L_{\epsilon,J}=1$. Increasing $J$ raises this ceiling but also introduces more pairwise reliability constraints, which can require more receive APs, higher sensing SNR, or longer sensing duration.
\section{Conclusion}
\label{sec:conclusion}

This paper established a theoretical foundation for networked embodied communication, showing how distributed sensing can make deliberately configured physical states reliably interpretable as network-level messages. 
The resulting analysis leads to several insights into how sensing resources should be understood in networked embodied communication.
\begin{itemize}[leftmargin=0.4cm]
    \item First, the most informative receive AP in isolation need not be the most valuable one to the network; what matters is the new distinction it contributes beyond the observations already available. With correlated observations, this contribution can emerge through cross-AP statistical relationships that are invisible from marginal observations alone.
    \item Second, spatial cooperation and temporal repetition play fundamentally different roles. Repetition can only accumulate evidence for distinctions that already exist, whereas additional viewpoints can reveal distinctions that are absent under the current receive AP set and thereby move the system from a regime with a nonzero error floor to one in which the decoding error vanishes asymptotically.
    \item Third, improving many symbol pairs does not necessarily improve alphabet-level reliability: once the alphabet is identifiable, the error exponent increases only when every currently limiting pair gains additional distinguishability.
\end{itemize}
These findings show that the value of a sensing network is determined not by the standalone quality or aggregate strength of its receivers, but by the complementary message distinctions their observations collectively provide.

\bibliographystyle{IEEEtran}
\bibliography{refs}

\clearpage
\section*{Supplementary Material}
The proofs omitted from the main paper are provided below in their order of appearance.

\subsection*{Proof of Theorem~\ref{thm:additional_ap_distinguishability}}
\begin{NewProof}
Factor the joint densities as
\[
p_k(\bm y,\bm z)
=
p_k(\bm y)
p_k(\bm z\mid\bm y),
\qquad
k\in\{i,j\}.
\]
The Chernoff coefficient over $\cS_2$ is
\[
\begin{aligned}
e^{-C_{ij}(s;\cS_2)}
&=
\int
p_i^s(\bm y)
p_j^{1-s}(\bm y)
\rho_s(\bm y)
\dd\bm y,\\
\rho_s(\bm y)
&=
\int
p_i^s(\bm z\mid\bm y)
p_j^{1-s}(\bm z\mid\bm y)
\dd\bm z.
\end{aligned}
\]
From \eqref{eq:conditional_innovation_model}, the two conditional Gaussian laws have means $\bm A_i\bm y$ and $\bm A_j\bm y$ and covariance matrices $\bm V_i$ and $\bm V_j$. Let $\Delta\bm A_{ij}=\bm A_i-\bm A_j$. Evaluating the Gaussian integral over $\bm z$ gives
\[
\begin{aligned}
\rho_s(\bm y)
&=
\exp[-\Delta_{ij}^{\rm cov}(s)]\\
&\quad\times
\exp\left[
-s(1-s)
\bm y^H
\Delta\bm A_{ij}^H
\overline{\bm V}_{ij}^{-1}(s)
\Delta\bm A_{ij}
\bm y
\right].
\end{aligned}
\]

Normalizing $p_i^s(\bm y)p_j^{1-s}(\bm y)$ by its integral gives a zero mean Gaussian tilted distribution with covariance
\[
\bm Q_{ij}(s)
=
\left[
s\bm\Sigma_{y,i}^{-1}
+
(1-s)\bm\Sigma_{y,j}^{-1}
\right]^{-1}.
\]
Let
\[
\bm T_s
=
s(1-s)\Delta\bm A_{ij}^H
\overline{\bm V}_{ij}^{-1}(s)\Delta\bm A_{ij}.
\]
Hence,
\[
\begin{aligned}
&\exp[-C_{ij}(s;\cS_2)+C_{ij}(s;\cS_1)]\\
&\quad=
\exp[-\Delta_{ij}^{\rm cov}(s)]
\E_{\widetilde p_s}\!\left[e^{-\bm y^H\bm T_s\bm y}\right],
\end{aligned}
\]
where $\widetilde p_s=\CN(\bm0,\bm Q_{ij}(s))$. Using
\[
\E\!\left[e^{-\bm y^H\bm T\bm y}\right]
=
\det(\bm I+\bm Q\bm T)^{-1}
\]
for $\bm y\sim\CN(\bm0,\bm Q)$ and $\bm T\succeq\bm0$, followed by Sylvester's determinant identity, gives
\[
\begin{aligned}
\E_{\widetilde p_s}\!\left[e^{-\bm y^H\bm T_s\bm y}\right]
&=
\frac{\det\overline{\bm V}_{ij}(s)}
{\det\bigl(\overline{\bm V}_{ij}(s)+\bm\Omega_{ij}(s)\bigr)}.
\end{aligned}
\]
Taking the negative logarithm yields \eqref{eq:chernoff_function_increment}.

Concavity of $\log\det(\cdot)$ gives $\Delta_{ij}^{\rm cov}(s)\ge0$, with equality exactly when $\bm V_i=\bm V_j$. Since $\bm Q_{ij}(s)\succ\bm0$ and $\overline{\bm V}_{ij}(s)\succ\bm0$, $\Delta_{ij}^{\rm mean}(s)\ge0$, with equality exactly when $\bm A_i=\bm A_j$. This proves \eqref{eq:conditional_redundancy_condition}.

If \eqref{eq:conditional_redundancy_condition} holds, the two Chernoff functions are identical for all $s$, so their optimized values are equal. Otherwise, the increment is strictly positive for every $s\in(0,1)$. When $C_{ij}^{\star}(\cS_1)>0$, its maximizer lies in $(0,1)$, and evaluating the larger Chernoff function at this maximizer gives $C_{ij}^{\star}(\cS_2)>C_{ij}^{\star}(\cS_1)$. When $C_{ij}^{\star}(\cS_1)=0$, any $s\in(0,1)$ gives $C_{ij}(s;\cS_2)>0$. This proves \eqref{eq:pairwise_chernoff_monotonicity} and its equality condition.

If $i\sim_{\cS_2}j$, marginalizing the observations in $\cS_2\setminus\cS_1$ gives $i\sim_{\cS_1}j$, so $\Pi_{\cS_2}$ refines $\Pi_{\cS_1}$. Conversely, for a pair satisfying $i\sim_{\cS_1}j$, we have $C_{ij}^{\star}(\cS_1)=0$. The equality condition above gives
\[
i\sim_{\cS_2}j
\quad\Longleftrightarrow\quad
\bm A_i=\bm A_j
\ \text{and}\
\bm V_i=\bm V_j.
\]
Combining the two directions proves \eqref{eq:identifiability_under_ap_expansion}. Applying this condition to every $i<j$ that is equivalent under $\cS_1$ gives the final statement.
\end{NewProof}

\subsection*{Proof of Corollary~\ref{cor:independent_identifiability}}
\begin{NewProof}
Under independence, the joint density is the product of the local AP densities. Two embodied symbols induce the same joint distribution if and only if their local distributions are identical at every selected receive AP. Hence, a pair is distinguishable jointly if and only if it belongs to $\mathcal D_r$ for at least one $r\in\cS$.
\end{NewProof}

\subsection*{Proof of Theorem~\ref{thm:reliability_characterization}}
\begin{NewProof}
For the nonidentifiable regime, consider an equivalence class $\mathcal E\in\Pi_{\cS}$ and denote its common $L$-interval observation distribution by $p_{\mathcal E}^{(L)}$. For any randomized decoder with output $\widehat W$,
\[
\begin{aligned}
&\sum_{i\in\mathcal E}
\Pr\{\widehat W=i\mid W=i\}\\
&\quad=
\int
p_{\mathcal E}^{(L)}\bigl(\bm Y^{(L)}\bigr)
\sum_{i\in\mathcal E}
\Pr\{\widehat W=i\mid\bm Y^{(L)}\}
\dd\bm Y^{(L)}\\
&\quad\le1.
\end{aligned}
\]
Hence at least one member of $\mathcal E$ has correct decision probability no greater than $1/|\mathcal E|$. Taking an equivalence class of size $\kappa_{\max}(\cC,\cS)$ gives $P_{\max}^{\star}(L;\cC,\cS)\ge1-1/\kappa_{\max}(\cC,\cS)$ for every $L$.

For achievability, regard each class in $\Pi_{\cS}$ as one statistical hypothesis and apply ML detection among the distinct class distributions. Any two distinct classes induce different observation distributions, so their pairwise Chernoff information is positive. Since the number of classes is finite, the probability of selecting the wrong class tends to zero as $L\to\infty$. After a class $\mathcal E$ is selected, choose one of its members uniformly. For any $i\in\mathcal E$, the conditional error probability therefore converges to $1-1/|\mathcal E|$. The largest asymptotic error occurs in a class of size $\kappa_{\max}(\cC,\cS)$, proving \eqref{eq:ambiguity_error_floor}.

For the identifiable regime, identifiability gives $C_{\min}^{\star}(\cC,\cS)>0$. Conditioned on $W=i$, an ML error implies that at least one competing likelihood is no smaller than the likelihood of $i$. The pairwise Chernoff bound and the union bound therefore give
\[
P_{\max}^{\rm ML}(L;\cC,\cS)
\le
(J-1)
\exp
\left[
-LC_{\min}^{\star}(\cC,\cS)
\right].
\]
Since $P_{\max}^{\star}(L;\cC,\cS)\le P_{\max}^{\rm ML}(L;\cC,\cS)$, the asymptotic error exponents of both quantities are at least $C_{\min}^{\star}(\cC,\cS)$.

For the converse, choose a pair $(i^{\star},j^{\star})$ attaining $C_{\min}^{\star}(\cC,\cS)$. Let $P_{e,ij}^{\star}(L)$ denote the minimum equal-prior binary error probability for discriminating embodied symbols $i$ and $j$ from $\bm Y^{(L)}$. Any $J$-ary decoder induces a binary decoder for the pair $(i^{\star},j^{\star})$ by deciding $i^{\star}$ only when the original decoder outputs $i^{\star}$ and deciding $j^{\star}$ otherwise. Under $W=i^{\star}$, the induced binary error probability is no greater than the original conditional error probability. Under $W=j^{\star}$, $\Pr\{\widehat W=i^{\star}\mid W=j^{\star}\}\le\Pr\{\widehat W\ne j^{\star}\mid W=j^{\star}\}$. Hence the average error probability of the induced binary decoder is no greater than the maximum error probability of the original decoder. It follows that $P_{\max}^{\star}(L;\cC,\cS)\ge P_{e,i^{\star}j^{\star}}^{\star}(L)$. The same construction applied to the ML decoder gives $P_{\max}^{\rm ML}(L;\cC,\cS)\ge P_{e,i^{\star}j^{\star}}^{\star}(L)$. By the binary Chernoff theorem \cite{chernoff1952measure,poor1994introduction},
\[
\lim_{L\to\infty}
-\frac{1}{L}
\log
P_{e,i^{\star}j^{\star}}^{\star}(L)
=
C_{i^{\star}j^{\star}}^{\star}(\cS)
=
C_{\min}^{\star}(\cC,\cS).
\]
Combining the upper and lower bounds proves \eqref{eq:maximum_error_exponents}.

We next consider $\emptyset\ne\cS_1\subset\cS_2$. Theorem~\ref{thm:additional_ap_distinguishability} shows that $\Pi_{\cS_2}$ refines $\Pi_{\cS_1}$ and that $C_{ij}^{\star}(\cS_2)\ge C_{ij}^{\star}(\cS_1)$ for every $i<j$.
A refinement cannot increase the size of the largest equivalence class, and taking the minimum of the pairwise inequalities gives \eqref{eq:reliability_monotonicity}.

For case (a), suppose $\cC$ is nonidentifiable under both $\cS_1$ and $\cS_2$. By part 1), the asymptotic error floor is $1-1/\kappa_{\max}(\cC,\cS)$. Since $1-1/\kappa$ is strictly increasing for integer $\kappa\ge2$, the error floor under $\cS_2$ is strictly smaller than that under $\cS_1$ if and only if $\kappa_{\max}(\cC,\cS_2)<\kappa_{\max}(\cC,\cS_1)$.

For case (b), suppose $\kappa_{\max}(\cC,\cS_1)>1$. Since $\Pi_{\cS_2}$ refines $\Pi_{\cS_1}$, every pair that is distinguishable under $\cS_1$ remains distinguishable under $\cS_2$. Hence, $\kappa_{\max}(\cC,\cS_2)=1$ if and only if every $i<j$ satisfying $i\sim_{\cS_1}j$ becomes distinguishable under $\cS_2$. By \eqref{eq:identifiability_under_ap_expansion}, this is equivalent to requiring $\bm A_i\ne\bm A_j$ or $\bm V_i\ne\bm V_j$ for every such pair. Parts 1) and 2) then show that the nonzero error floor under $\cS_1$ is replaced by exponential error decay under $\cS_2$.

For case (c), suppose $\cC$ is identifiable under $\cS_1$ and let $c_1=C_{\min}^{\star}(\cC,\cS_1)$. By pairwise monotonicity, every pair satisfying $C_{ij}^{\star}(\cS_1)>c_1$ remains strictly above $c_1$ under $\cS_2$. For a pair attaining $c_1$, Theorem~\ref{thm:additional_ap_distinguishability} shows that its Chernoff information strictly increases if and only if $\bm A_i\ne\bm A_j$ or $\bm V_i\ne\bm V_j$. Hence, if every pair attaining $c_1$ satisfies this condition, all finitely many symbol pairs have Chernoff information strictly larger than $c_1$ under $\cS_2$, so $C_{\min}^{\star}(\cC,\cS_2)>c_1$. Conversely, if one pair attaining $c_1$ satisfies $\bm A_i=\bm A_j$ and $\bm V_i=\bm V_j$, its Chernoff information remains equal to $c_1$, and pairwise monotonicity gives $C_{\min}^{\star}(\cC,\cS_2)=c_1$. This proves case (c) and completes the proof.
\end{NewProof}

\subsection*{Proof of Proposition~\ref{prop:maximum_error_bound}}
\begin{NewProof}
Condition on $W=i$. An ML error occurs only if at least one competing hypothesis $j\ne i$ has likelihood no smaller than that of $i$. For each competing pair, the Chernoff bound over $L$ independent sensing intervals gives
\[
\begin{aligned}
&\Pr\!\left\{p_j^{(L)}(\bm Y^{(L)})\ge p_i^{(L)}(\bm Y^{(L)})\mid W=i\right\}\\
&\qquad\le \exp\!\left[-LC_{ij}^{\star}(\cS)\right].
\end{aligned}
\]
Applying the union bound over the $J-1$ competing hypotheses yields
\[
\begin{aligned}
&\Pr\{\delta_{\rm ML}(\bm Y^{(L)})\ne i\mid W=i\}\\
&\quad\le
\sum_{j\ne i}
\exp\!\left[-LC_{ij}^{\star}(\cS)\right]\\
&\quad\le
(J-1)
\exp\!\left[-LC_{\min}^{\star}(\cC,\cS)\right].
\end{aligned}
\]
Taking the maximum over $i$ proves \eqref{eq:maximum_error_bound}. Requiring the right-hand side to be no greater than $\epsilon$ gives \eqref{eq:reliability_requirement}.
\end{NewProof}

\subsection*{Proof of Proposition~\ref{prop:pairwise_necessary_condition}}
\begin{NewProof}
Choose any pair $(i,j)$. From the original $J$-ary decoder, construct a binary decoder that decides $i$ only when the original output is $i$ and decides $j$ otherwise. Under $W=i$, its error probability is no greater than the original conditional error probability. Under $W=j$, $\Pr\{\widehat W=i\mid W=j\}\le\Pr\{\widehat W\ne j\mid W=j\}$. Hence its equal-prior average error probability is no greater than $\epsilon$, and therefore $P_{e,ij}^{\star}(L)\le\epsilon$.

For equal priors, $P_{e,ij}^{\star}(L)=[1-\operatorname{TV}(p_i^{(L)},p_j^{(L)})]/2$, where $\operatorname{TV}(p,q)=\frac{1}{2}\int|p(\bm u)-q(\bm u)|\dd\bm u$.
Thus, $\operatorname{TV}(p_i^{(L)},p_j^{(L)})\ge1-2\epsilon$. The total variation distance and the Bhattacharyya coefficient satisfy \cite{poor1994introduction}
\[
\operatorname{TV}
\bigl(
p_i^{(L)},
p_j^{(L)}
\bigr)
\le
\sqrt{
1-
\left(
\int
\sqrt{
p_i^{(L)}p_j^{(L)}
}
\right)^2
}.
\]
Combining the two inequalities gives $\bigl(\int\sqrt{p_i^{(L)}p_j^{(L)}}\bigr)^2\le4\epsilon(1-\epsilon)$. The Bhattacharyya distance of $L$ independent observations is $LB_{ij}$, so $LB_{ij}\ge\frac{1}{2}\log\frac{1}{4\epsilon(1-\epsilon)}$.

The Bhattacharyya requirement can now be converted into a necessary condition on the optimized Chernoff information. Since $C_{ij}^{\star}(\cS)\ge C_{ij}(1/2;\cS)=B_{ij}$, every pair must satisfy $C_{ij}^{\star}(\cS)\ge\frac{1}{2L}\log\frac{1}{4\epsilon(1-\epsilon)}$.
Taking the minimum over all pairs proves \eqref{eq:necessary_reliability_requirement}.
\end{NewProof}

\subsection*{Proof of Proposition~\ref{prop:ap_objective_structure}}
\begin{NewProof}
For nonempty $\cS_1\subseteq\cS_2$, monotonicity follows from Theorem~\ref{thm:additional_ap_distinguishability}, which gives
\[
C_{ij}^{\star}(\cS_2)
\ge
C_{ij}^{\star}(\cS_1),
\qquad
\cS_1\subseteq\cS_2,
\]
for every symbol pair.
Taking the minimum over all pairs gives $f_{\cC}(\cS_2)\ge f_{\cC}(\cS_1)$.
The case $\cS_1=\emptyset$ follows from $f_{\cC}(\emptyset)=0$ and the nonnegativity of Chernoff information.

To show that diminishing returns need not hold, consider three embodied symbols and two independent receive APs with scalar observation variances
\[
\begin{array}{c|ccc}
&\bm q_1&\bm q_2&\bm q_3\\
\hline
r_1&\nu_{\rm A}&\nu_{\rm B}&\nu_{\rm B}\\
r_2&\nu_{\rm A}&\nu_{\rm A}&\nu_{\rm B}
\end{array},
\]
where $\nu_{\rm A}\ne\nu_{\rm B}$ are positive.
Let $C_0>0$ denote the Chernoff information between $\CN(0,\nu_{\rm A})$ and $\CN(0,\nu_{\rm B})$.

With only $r_1$, symbols $\bm q_2$ and $\bm q_3$ are equivalent, while with only $r_2$, symbols $\bm q_1$ and $\bm q_2$ are equivalent.
Hence,
\[
f_{\cC}(\emptyset)
=
f_{\cC}(\{r_1\})
=
f_{\cC}(\{r_2\})
=
0.
\]
Because the two nonzero local binary problems use the same variance pair $(\nu_{\rm A},\nu_{\rm B})$, their Chernoff functions have the same maximizing parameter.
Under the joint observation,
\[
C_{12}^{\star}
=
C_0,
\qquad
C_{23}^{\star}
=
C_0,
\qquad
C_{13}^{\star}
=
2C_0.
\]
Thus, $f_{\cC}(\{r_1,r_2\})=C_0$.
The gain from adding $r_2$ to the empty set is therefore zero, whereas the gain from adding it after $r_1$ is $C_0$.
This violates the diminishing returns condition required for submodularity.
\end{NewProof}

\subsection*{Proof of Proposition~\ref{prop:signal_subspace_reduction}}
\begin{NewProof}
For $k\in\{i,j\}$, \eqref{eq:block_cov} and the definition of $\bm U_k$ give $\bm K^{-1/2}\bm\Sigma_k\bm K^{-1/2}=\bm I+\bm U_k\bm U_k^H$. Taking determinants yields $\det\bm\Sigma_k=\det(\bm K)\det(\bm I+\bm U_k\bm U_k^H)$.
Similarly,
\[
\begin{aligned}
&\bm K^{-1/2}
\bigl[
(1-s)\bm\Sigma_i+s\bm\Sigma_j
\bigr]
\bm K^{-1/2}\\
&\qquad=
\bm I+(1-s)\bm U_i\bm U_i^H
+s\bm U_j\bm U_j^H\\
&\qquad=
\bm I+
\bm U_{ij}(s)\bm U_{ij}^H(s).
\end{aligned}
\]
Thus, $\det\bigl((1-s)\bm\Sigma_i+s\bm\Sigma_j\bigr)=\det(\bm K)\det\bigl(\bm I+\bm U_{ij}(s)\bm U_{ij}^H(s)\bigr)$. Substituting these determinant expressions into \eqref{eq:gaussian_chernoff} cancels the $\log\det\bm K$ terms. Applying Sylvester's determinant identity, $\det(\bm I+\bm U\bm U^H)=\det(\bm I+\bm U^H\bm U)$, to the three remaining determinants gives \eqref{eq:reduced_chernoff}.
\end{NewProof}

\end{document}